\documentclass[cleveref,autoref,thm-restate]{lipics-v2021}
\RequirePackage{amsmath, amssymb, graphics, graphicx,amsthm,mathtools}

\nolinenumbers
\RequirePackage[T1]{fontenc}
\RequirePackage{subcaption,complexity} 
\RequirePackage{wrapfig}
\RequirePackage{subcaption}
\RequirePackage{todonotes}
\RequirePackage{comment}
\RequirePackage{soul}
 \RequirePackage{makecell}
\RequirePackage{todonotes}
\RequirePackage{array}
\RequirePackage{color}
\RequirePackage{cases}
\RequirePackage{xparse}
\RequirePackage{xargs}
\RequirePackage{appendix}
\RequirePackage{xcolor}
\RequirePackage{relsize}
\RequirePackage{xspace}
\RequirePackage{tabularx}
\usepackage{boxedminipage}

\renewcommand{\deg}{\ensuremath{\mathrm{deg}}}

\theoremstyle{remark}

\newtheorem{algorithm}[theorem]{Algorithm}
\newtheorem{reduction}[theorem]{Reduction}
\crefname{observation}{Observation}{Observations}
\Crefname{observation}{Observation}{Observations}

  \usepackage{import}
\hideLIPIcs

\newcommand{\sdmstp}{\textsc{Specified Degree MST}\xspace}
\newcommand{\setmstp}{\textsc{Set of Degrees MST}\xspace}

\newcommand{\weightf}{\ensuremath{\omega}}
 \DeclareMathOperator{\ctw}{ctw}
\DeclareMathOperator{\pw}{pw}
\DeclareMathOperator{\tw}{tw}
\DeclareMathOperator{\cw}{cw}
\DeclareMathOperator{\td}{td}
\DeclareMathOperator{\vcn}{vcn}
\DeclareMathOperator{\rk}{rk}

 \newcommand{\req}{D}

\usepackage{drawings}

\graphicspath{{./graphics/}}

 \usepackage[most]{tcolorbox}
\usepackage{xcolor}

\newtcolorbox{nicequote}{
  enhanced,
  colback=gray!5,
  colframe=gray!50,
  boxrule=0.5pt,
  arc=2pt,
  left=3em,      right=1em,
  top=1em,
  bottom=1em,
  overlay={
    \node[
      anchor=north west,
      text=gray!45,
      font=\fontsize{40}{40}\selectfont
    ] at ([xshift=0.35em,yshift=-0.25em]frame.north west) {\textbf{``}};
  }
}

 \title{Not All Degree Constraints Are Created Equal\\ when Computing Spanning Trees}
\titlerunning{Not All Degree Constraints Are Created Equal when Computing Spanning Trees}  \author{Narek Bojikian}{Humboldt Universität zu Berlin, Germany}{bojikian@hu-berlin.de}{https://orcid.org/0000-0003-1072-4873}{}
\author{Alexander Firbas}{TU Wien, Austria}{alexander.firbas@tuwien.ac.at}{https://orcid.org/0009-0007-2049-2144}{FWF Project 10.55776/Y1329 and WWTF Project 10.47379/ICT22029}
\author{Robert Ganian}{TU Wien, Austria}{rganian@gmail.com}{https://orcid.org/0000-0002-7762-8045}{FWF Project 10.55776/Y1329 and WWTF Project 10.47379/ICT22029}
\author{Hung P. Hoang}{TU Wien, Austria}{phoang@ac.tuwien.ac.at}{https://orcid.org/0000-0001-7883-4134}{FWF Projects 10.55776/Y1329 and ESP1136425}
\author{Krisztina Szil\'{a}gyi}{Czech Technical University, Czechia}{krisztina.szilagyi@fit.cvut.cz}{https://orcid.org/0000-0003-3570-0528}{Supported under the project Robotics and advanced industrial production (reg. no. CZ.02.01.01/00/22\_008/0004590) and CTU Global Postdoc Fellowship Program.}

\Copyright{Narek Bojikian, Alexander Firbas, Robert Ganian, Hung P. Hoang, Krisztina Szil\'{a}gyi}

\authorrunning{N. Bojikian, A. Firbas, R. Ganian, H. P. Hoang, K. Szil\'{a}gyi}



\begin{document}

\maketitle

\begin{abstract}
We study the computation of minimum spanning trees subject to local degree constraints. Recent work (ICALP 2026) established that three natural formalizations of this problem share the exact same parameterized complexity under standard structural graph parameters, including treewidth, pathwidth and clique-width. This applies to the cases where every vertex has a single target degree (\textsc{Specified Degree MST}), or a degree upper bound (\textsc{Bounded Degree MST}), or is equipped with a set of admissible degrees (\textsc{Set of Degrees MST}).

In this paper, we investigate these problems under more restrictive parameterizations and reveal that their complexity landscapes fundamentally diverge on bounded-treedepth graphs. Specifically, we prove that the former two problems are fixed-parameter tractable when parameterized by the treedepth of the input graph. In sharp contrast, we show that \textsc{Set of Degrees MST} remains $\W[1]$-hard parameterized by treedepth, even when combined with the feedback vertex number (i.e., deletion distance to treewidth $1$). Finally, we show that this divergence seems to be specific to treedepth: we exclude an analogous $\W[1]$-hardness result for \textsc{Set of Degrees MST} w.r.t.\ the vertex cover number and also rule out fixed-parameter algorithms for the former two problems w.r.t.\ deletion distance to constant pathwidth.
\end{abstract}

\section{Introduction}
\label{sec:intro}
The computation of minimum-cost spanning trees is a foundational topic in computer science, with classical methods like Prim's and Kruskal's algorithms serving as staple introductions to graph theory. Yet, in some settings we require spanning trees which possess additional properties (beyond connectivity and minimum weight). 
 In this article, we build on the recent work of Bojikian, Firbas, Ganian, Hoang and Szil\'agyi who investigated three natural formalizations of the task of finding minimum-cost spanning trees which satisfy prescribed degree constraints~\cite{BFGHS26}:

\begin{nicequote}
    Given a graph $G$ with polynomially-bounded edge weights and a constraint function $\req \colon V(G)\rightarrow 2^{\mathbb{N}}$, our aim is to determine whether there is a spanning tree $T$ of $G$ such that for each $v\in V(G)$, $\deg_T(v)\in \req(v)$---and if the answer is positive, output one of minimum cost. Depending on the form of the constraint function, we distinguish between the following three computational problems:		
    \begin{enumerate}
    \item in \textsc{Set of Degrees Minimum Spanning Tree (MST)}, $\req$ maps each vertex to a set of integers;
    \item in \textsc{Bounded Degree MST}, all sets in the image of $\req$ are of the form $\{1,\dots,d\}$ for some $d\in\mathbb{N}$;
    \item in \textsc{Specified Degree MST}, the image of $\req$ is a set of singletons.
\end{enumerate}
\end{nicequote}

Although these problems form a hierarchy from most to least general\footnote{\textsc{Specified Degree MST} can easily be reduced to \textsc{Bounded Degree MST} without changing $G$~\cite{BFGHS26}.}, each constraint type is fundamentally important. \textsc{Bounded Degree MST} serves as a crucial base case for the \emph{Thin Tree Conjecture}~\cite{Goddyn,DBLP:conf/focs/KleinO23} and is widely studied in approximation settings~\cite{FurerR92,DBLP:conf/focs/Goemans06,SinghL15}. Meanwhile, \textsc{Specified} and \textsc{Set of Degrees MST} act as natural counterparts to the classical \textsc{$f$-Factor} and \textsc{General $f$-Factor} problems~\cite{tutte1952factors,Cornuejols88,GabowS21,GabowS21a}, whose \NP-hard connected variants are also well-established~\cite{ellingham2002connected,CornelissenHMNR18,GanianNORR19}.

The aforementioned preceding paper~\cite{BFGHS26} established a comprehensive classification of fine-grained running times for these degree-constrained spanning tree problems across various structural graph parameters under the Exponential Time Hypothesis (ETH) and its strong variant (SETH)~\cite{ImpagliazzoP01,ImpagliazzoPZ01}. Remarkably, the authors demonstrated that all three problems share the exact same complexity for each parameter considered in their work: their lower bounds apply to the most restrictive \textsc{Specified Degree MST}, while their algorithmic upper bounds accommodate the most general \textsc{Set of Degrees MST}. 
 For clique-width ($\cw$), they provided an ETH-tight $n^{\mathcal{O}(\cw)}$-time algorithm on unweighted graphs. For treewidth ($\tw$) and pathwidth ($\pw$), they established single-exponential \XP-algorithms along with an (almost-matching) SETH lower bound which, among others, excludes fixed-parameter tractability.
 Finally, they proved that when parameterized by cutwidth ($\ctw$), the problems admit a SETH-tight $\mathcal{O}^*(3^{\ctw(G)})$-time algorithm.

\subparagraph{Contributions.} 
 In this article, we revisit these problems under the lens of more restrictive parameterizations than treewidth. In particular, we focus on \emph{treedepth}, a fundamental structural parameter that characterizes how close a graph is to a star forest. Perhaps the most intuitive and algorithmically useful characterization of the treedepth $\td(G)$ of a graph $G$ is via recursive vertex elimination: it is exactly the minimum depth of a recursion tree in which, at each step, one deletes a single vertex and recurses independently on the newly formed connected components until each such component is a singleton. Because graphs with bounded treedepth cannot contain long paths, bounded treedepth constitutes a strictly stronger restriction than bounded treewidth. This---along with its fundamental ties to graph sparsity~\cite{sparsitybook} and space-efficient algorithms~\cite{PilipczukW18,NederlofPSW23}, including more recent breakthroughs on fine-grained logic-based meta theorems~\cite{DBLP:conf/soda/BergougnouxCS26}---makes it a natural parameterization that has been successfully employed for a variety of problems. 
 
As our first result, we show that treedepth can indeed be leveraged to push the boundaries of tractability for some variants of degree-constrained spanning trees:

\begin{restatable}{theorem}{thmtd}
\label{thm:td}
\textsc{\textup{Bounded Degree MST}} is in \FPT\ w.r.t.\ the treedepth of~$G$.
\end{restatable}

The proof of Theorem~\ref{thm:td} follows by encoding the problem into an Integer Linear Program (ILP) which is ``well-structured'': both the multiplicative coefficients and the treedepth of the ILP's dual graph are bounded by a function of $\td(G)$.
 It is known that such instances of ILP admit a fixed-parameter algorithm~\cite[Theorem 6]{treedepthilp}; however, encoding acyclicity and connectivity in such an ILP is challenging and our approach relies on insights into the structure of hypothetical solutions and a notion of rank inspired by graphic matroids.
 
With the above theorem in hand and recalling the previous results on the three problems, one could hope for a way of lifting tractability also to \textsc{Set of Degrees MST}---either via a more refined ILP encoding or via different techniques entirely. A natural stepping stone in this direction would be to attack instances where each vertex must have one of two possible degrees in the spanning tree. Surprisingly, we show that such an algorithm is impossible under current complexity assumptions, and this holds even on graphs which are ``almost trees'' (in the sense of having small feedback vertex number). This means that---unlike for all parameters studied in the preceding work---\textbf{which problem variant we target actually matters on bounded-treedepth graphs}.
 
\begin{restatable}{theorem}{thmtdhard}
\label{thm:tdhard}
\textsc{\textup{Set of Degrees MST}} is \W\textup{[1]}-hard w.r.t.\ the treedepth plus the feedback vertex number of $G$, even on unweighted graphs with $|\req(v)|\leq 2$ for each vertex $v$.
\end{restatable}

At the heart of the proof of Theorem~\ref{thm:tdhard} lies a carefully constructed gadget that can model a choice between either summing up two $k$-dimensional integer vectors or keeping one of the vectors intact. This gadget then allows us to reduce from a known \W[1]-hard variant of the \textsc{Multidimensional Subset Sum} problem~\cite{GanianOR23}.

We complement the above results by showing that the diverging behavior between \textsc{Set of Degrees MST} and the ``simpler'' two problem variants seems to occur specifically on treedepth. On one hand, when parameterized by the vertex cover number---a more restrictive parameterization than treedepth which merely measures the vertex deletion distance to singleton components---  we can obtain tractability for the former case at least on unweighted graphs, by reducing the problem to a generalization of (B-)matching~\cite{DBLP:journals/algorithmica/GutinKSSY12}:

\begin{restatable}{theorem}{thmvcfpt}
\label{thm:vcfpt}
On unweighted graphs, \textsc{\textup{Set of Degrees MST}} is fixed-parameter tractable w.r.t.\ the vertex cover number.
\end{restatable}

On the other hand, as \textsc{Set of Degrees MST} is \W[1]-hard w.r.t.\ the feedback vertex number (which is the vertex deletion distance to treewidth $\leq 1$), one might wonder whether the vertex deletion distance to bounded-treewidth  or bounded-pathwidth graphs is an alternative structural measure where the parameterized complexity of the problems differ. We show that this is not the case via a separate lower bound:

\begin{restatable}{theorem}{thmpwdeletion}
\label{thm:pwdeletion}
\textsc{\textup{Specified Degree MST}} is \W\textup{[1]}-hard when parameterized by the vertex deletion distance to graphs of pathwidth $\leq 4$ and treewidth $\leq 3$, even on unweighted graphs.
\end{restatable}

                An overview of our results is provided in Table~\ref{tbl:new-results}.

\begin{table}
\centering
\small
\begin{tabular}{c|c|c}
Parameter & \sdmstp & \setmstp \\
\hline
Treedepth & FPT (Thm.~\ref{thm:td}) & $W[1]$-hard (Thm.~\ref{thm:tdhard}) \\
Vertex cover number & FPT (Thm.~\ref{thm:td}) & FPT$^\ast$ (Thm.~\ref{thm:vcfpt}) \\
Deletion distance to $\pw \leq 4$ & $W[1]$-hard (Thm.~\ref{thm:pwdeletion}) & $W[1]$-hard (Thm.~\ref{thm:pwdeletion}) \\
\end{tabular}
\smallskip
\caption{Overview of our results for stronger structural restrictions. The vertex-cover algorithm marked by $^\ast$ is for edge-unweighted instances.}
\label{tbl:new-results}
\end{table}

\subparagraph{Paper Organization.} In~\cref{sec:prelim}, we introduce notation that will be used throughout the paper. In~\cref{sec:fpt_td}, we describe an FPT algorithm for \textsc{Bounded Degree MST} parameterized by treedepth. \Cref{section:fvs_w1} shows $W[1]$-hardness of \setmstp parameterized by feedback vertex number plus treedepth. In~\cref{sec:w_deletion}, we prove that \sdmstp is $W[1]$-hard w.r.t. vertex deletion distance to graphs of constant pathwidth. In~\cref{sec:fpt_vc} we show that \setmstp is in FPT parameterized by vertex cover on unweighted graphs. We conclude with possible directions for further research in~\cref{sec:conclusion}.

\section{Preliminaries}
\label{sec:prelim}
For a graph $G$ and a vertex $v$ of $G$, we denote by $\deg_G(v)$ the degree of $v$ in $G$.
We denote by $N_G(v)$ the set of neighbors of $v$ in $G$ (i.e., $N_G(v)=\{u\in V(G):\: uv\in E(G)\}$).
We omit the subscript $G$, when it is clear from context.
For a vertex set $S\subseteq V(G)$, we denote by $G[X]$ the subgraph of $G$ induced by $X$, and for two disjoint subsets $X,Y\subseteq V(G)$, we denote by $G[X,Y]$ the bipartite subgraph of $G$ induced by $X$ and $Y$.

We write $\mathbb N_0=\mathbb N\cup\{0\}$.
For $n\in\mathbb N_0$, let $[n]=\{1,\ldots,n\}$, and let $[n]_0=\{0,\ldots,n\}$.
For a multiset $M$ and an element $x$ of $M$, we denote by $\#_M(x)$ the multiplicity of $x$ in $M$.
Given a set $A\subseteq \mathbb{N}_0$ and some integer $x$, we denote by $A-x$ the set $\{a-x \mid a \in A\}$.

\subparagraph{Problem Definitions.} In \textsc{Set of Degrees MST}, we are given as input a triple $(G,\req,\weightf)$, where $G$ is a graph with edge weights $\weightf:E(G)\to\mathbb Z$ and $\req:V(G)\to 2^{\mathbb N_0}$. We are asked to find a spanning tree $T$ of $G$ of minimum weight such that for each $v\in V(G)$, $\deg_T(v)\in\req(v)$.
\textsc{Bounded Degree MST} and \textsc{Specified Degree MST} are defined analogously. For these two problems, in a slight abuse of notation, we let $\req$ denote an integer-valued function. In \textsc{Bounded Degree MST}, this means we require $\deg_T(v)\leq\req(v)$, and in \textsc{Specified Degree MST}, this means we require $\deg_T(v)=\req(v)$.
If $\weightf(e)=1$ for all edges $e\in E(G)$, then the instance is \emph{unweighted}; in this case, we omit the weight function.

\subparagraph{Graph Parameters.} Given a rooted tree $R$ with root $r$, we say $u$ is an ancestor of $v$ if it lies on the path from $v$ to $r$.
    An \emph{elimination tree} of a connected graph $G$ is a rooted tree $R$ with vertex set $V(R)=V(G)$ such that for every edge $uv\in E(G)$, $u$ is an ancestor of $v$ or $v$ is an ancestor of $u$ in $R$. The \emph{treedepth} of $G$ is the smallest possible depth of an elimination tree of $G$.

The \emph{vertex cover number} of a graph $G$ is the size of a minimum vertex cover in $G$.
The \emph{feedback vertex number} of $G$ is the minimum number of vertices whose removal from $G$ yields an acyclic graph (i.e., a forest).

 \section{FPT by Treedepth}
\label{sec:fpt_td}
\looseness=-1
In this section, we show that \textsc{Bounded Degree MST} admits a fixed-parameter algorithm parameterized by treedepth in the setting where edges can have integer weights.
We first model the problem as an ILP with bounded coefficient size and bounded treedepth of the dual graph, and then solve it with an algorithm by Kouteck{\'{y}}, Levin, and Onn~\cite[Theorem 6]{treedepthilp}.

\subsection{Setup for the Algorithm}
Throughout this section, we fix an instance of \textsc{Bounded Degree MST}, i.e., a connected graph $G$, degree-requirement $\req\colon V(G)\to\mathbb N_0$, edge-weights $\weightf:E(G)\to\mathbb Z$.
Further, we assume a rooted elimination
tree $R$ of $G$ of height $k$ with root $r$ is given.

For $v\in V(G)$, let
\[
    A_{v}^{-}=\{a\in V(G):a\text{ is a proper ancestor of }v\text{ in }R\},
    \qquad
    A_{v}^{+}=A_{v}^{-}\cup\{v\}.
\]
\[
    D_{v}^{-}=\{u\in V(G):u\text{ is a proper descendant of }v\text{ in }R\},
    \qquad
    D_{v}^{+}=D_{v}^{-}\cup\{v\}.
\]
Note that $v$ is also an ancestor/descendant of itself.
For a finite set $X$, let $\mathcal P(X)$ be the set of all partitions of $X$.
We call the elements of a partition blocks.

Each edge of $G$ is \emph{owned} by its deeper endpoint w.r.t. the elimination tree $R$.
For $v\in V(G)$, let
\[
    O(v)=\{va:a\in A_{v}^{-}\cap N_G(v)\},
    \qquad
    \mathcal S_v=2^{O(v)}.
\]
Thus, $O(v)$ is the set of edges owned by $v$, and $\{ O(v) \mid v \in V(G) \}$ forms a partition of $E(G)$.
In the ILP, each vertex $v$ will choose one subset of its owned edges, i.e., one $S \in \mathcal{S}_v$, to be part of the spanning tree.

For $v\in V(G)$, let
\[
    E_{v}=\bigcup_{u\in D_{v}^{+}}O(u),
\]
and let $G_{v}$ be the edge-induced subgraph of $G$ with edge set $E_{v}$.
Thus $G_v$ models the subgraph ``at or below'' $v$ in the elimination tree,
where edges can also ``reach outside'' to ancestors $A^-_v$.
The job of the ILP is now to ensure that the selected edges for $G_v$ for each $v$ 
form a spanning forest of $G_v$ (resp.\ a spanning tree for $v = r$).
Towards this, we track a \emph{connectivity partition} for each $G_v$,
that is, we record into which connected components the spanning forest in $G_v$ connects the ancestors of $v$, $A^+_v$.
At the root, where $G_r = G$, and $A^+_r = \{ r \}$, the only suitable partition is $\{ \{ r \} \}$.
At non-root vertices, we allow potentially all partitions of $A^+_v$, except partitions that include the block $\{v\}$, since selecting such a partition would render $v$ isolated in the selected graph.
Thus, we define for each $v\in V(G)$,
\[
    \mathcal P_{v}^{+}=
    \begin{cases}
        \bigl\{\{\{v\}\}\bigr\}, & v=r,\\[1mm]
        \{\Pi\in\mathcal P(A_{v}^{+}):\{v\}\notin\Pi\}, & v\ne r.
    \end{cases}
\]
For the remainder of this subsection, we introduce algebraic tools that will be important for establishing the correctness of the ILP encoding. 
 For a graph $H$ and a finite set $X$ of vertices,
we write $\Pi(H, X)$ for the connectivity partition induced by $H$ on $X$, that is,
two vertices $u,v \in X$ are in the same block $B \in \Pi(H, X)$ if and only if
$u$ and $v$ are in a shared connected component in $(V(H)\cup X,E(H))$. In particular, vertices in $X\setminus V(H)$ are singleton blocks.

For a partition $\Pi$, let the \emph{rank} of $\Pi$ be defined as 
\[
    \rk(\Pi)=\left|\bigcup_{B\in\Pi}B\right|-|\Pi|.
\]
Note that, mirroring the notion of rank in the graphic matroid,
for a connectivity partition $\Pi$ over a set $X$ induced by some set of edges $E_X$,
$\rk(\Pi)$ is the number of edges 
in any forest on $X$ whose connected
components are exactly the blocks of $\Pi$.
Consequently, a graph on $X$ whose connected components are exactly the blocks of $\Pi$
with more than $\rk(\Pi)$ edges contains a cycle.

For a partition $\Pi$ over a set $X$, we say that a forest $F$ with vertex set $X$ realizes $\Pi$ if the vertex sets of the connected components of $F$ are
exactly the blocks of $\Pi$.
\begin{lemma}\label{obs:rank_edge_lower_bound}
    Let $H$ be a graph or a loopless multigraph and let $X$ be a finite set of vertices.
    Then $\rk(\Pi(H,X))\le |E(H)|$.
	    Moreover, if a partition $\Pi$ is realized by a forest $F$, then
	    $\rk(\Pi)=|E(F)|$.
\end{lemma}
\begin{proof}
    Let us first prove the second claim. Let $F$ be a forest with vertex set $X$ realizing a partition $\Pi$ of a set $X$. The number of edges of a forest equals the number of its vertices minus the number of its connected components, so $\rk(\Pi)=|E(F)|$.

    \looseness=-1
    To prove the first claim, let $F$ be a forest realizing $\Pi(H, X)$.
    Note that, apart from singleton components of $F$ corresponding to vertices of $X\setminus V(H)$, there is a one-to-one correspondence between the connected components of $F$ and those of $H$ that contain a vertex of $X$, where edge multiplicities are ignored if $H$ is a multigraph.
    Let $D_1, \dots, D_t$ be the connected components of $H$ that contain a vertex of $X$, and $C_1, \dots, C_t$ be the corresponding connected components of $F$.
    For each $i \in [t]$, note that $D_i$ contains all the vertices of $C_i$, and $C_i$ is a forest (in fact, a tree) which, together with isolated vertices for $X\setminus V(C_i)$, realizes $\Pi(D_i,X)$.
    By the second claim of this lemma, we have $\rk(\Pi(D_i,X)) = |V(C_i)| - 1$.
    Combining the above with the fact that $D_i$ is connected, we obtain $|E(D_i)| \geq |V(D_i)| - 1 \geq |V(C_i)| - 1 = |E(C_i)|$.
    Summing over all $i \in [t]$, we then have $|E(H)| \geq \sum^{t}_{i=1} |E(D_i)| \geq \sum^{t}_{i=1} |E(C_i)| = |E(F)|$. 
    Again by the second claim of this lemma, we have $|E(F)| = \rk(\Pi(H,X))$.
    The lemma then follows.
\end{proof}

See \Cref{figure:treedepth_ilp_example} for an example of $G_v$, a spanning forest selected for $G_v$, 
and the resulting connectivity partition of $A^+_v$.

\begin{figure}[!htbp]
    \centering
    \includegraphics[page=1]{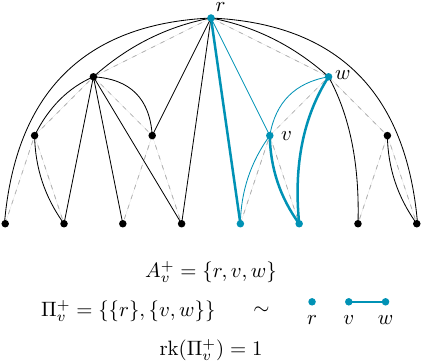}
   \caption{A graph $G$ (black and blue edges), drawn along an elimination tree of height $3$ rooted at $r$; the tree edges are shown as dashed gray edges. The subgraph $G_v$ of $G$ is drawn in blue, and the spanning forest selected for $G_v$ is shown in bold blue. The ancestors $A_v^+$ of $v$ are $r, v, w$, and $\Pi_v^+$ is the connectivity partition of $A_v^+$ induced by the selected spanning forest for $G_v$. This connectivity partition is realized by a spanning forest of $A_v^+$ with one edge, shown in blue. Since it has one edge, $\rk(\Pi_v^+) = 1$.}
    \label{figure:treedepth_ilp_example}
\end{figure}

Let $\Pi_1, \Pi_2$ be two partitions over the same set, i.e., $\bigcup_{B \in \Pi_1} B = \bigcup_{B \in \Pi_2} B$.
Then, if for every $B \in \Pi_1$, there is $B' \in \Pi_2$ with $B \subseteq B'$, we say $\Pi_1$ is \emph{finer} than $\Pi_2$ and write $\Pi_1 \preceq \Pi_2$.

For a partition $\rho$, a set $B$, and a nonempty proper subset
$U\subsetneq B$, let
\[
    \chi(\rho,U,B)=
    \begin{cases}
        1, & \exists B'\in\rho:
             B'\cap U\ne\emptyset\text{ and }B'\cap(B\setminus U)\ne\emptyset,\\
        0, & \text{otherwise}.
    \end{cases}
\]
That is, $\chi(\rho,U,B)$ is one if and only if $\rho$ contains a block that ``touches'' both sides of the cut $(U, B \setminus U)$.

For two partitions $\Pi_1,\Pi_2$ over a common set $X$, let $F_1,F_2$ be
forests realizing $\Pi_1,\Pi_2$, respectively. We define the \emph{join of $\Pi_1$ and $\Pi_2$} as 
$\Pi_1 \oplus \Pi_2=\Pi(F_1\cup F_2,X)$.
This is well-defined, as the right-hand side is independent of the chosen
realizing forests.
Equivalently, $\Pi_1\oplus\Pi_2$ is the finest partition $\Pi$ such that
$\Pi_1\preceq\Pi$ and $\Pi_2\preceq\Pi$.
Note that $\oplus$ is associative and that we use the convention that a one-term join is the partition itself.
We close this subsection with some elementary observations.

\begin{lemma}\label{obs:rank_monotone}
    Let $\Pi_1, \Pi_2$ be two partitions over the same finite set.
    If $\Pi_1 \preceq \Pi_2$, then $\rk(\Pi_1) \leq \rk(\Pi_2)$.
\end{lemma}
\begin{proof}
    Note that, by definition, $\cup_{B\in \Pi_1} B=\cup_{B\in \Pi_2} B$, so the statement is equivalent to showing $|\Pi_1|\geq |\Pi_2|$. To show this, we will construct an injective mapping $\varphi: \Pi_2\rightarrow \Pi_1$ as follows. For $B\in \Pi_2$, let $\varphi(B)$ be an arbitrary set in $\Pi_1$ such that $\varphi(B)\subseteq B$ (such a set exists as $\Pi_1$ is a partition). As $\Pi_2$ is a partition, this mapping is injective, so $|\Pi_1|\geq |\Pi_2|$.  
\end{proof}
\begin{lemma}\label{obs:rank_subadditive_join}
    Let $\Pi_1,\ldots,\Pi_t$ be partitions over the same finite set. Then
    \[
        \rk(\Pi_1\oplus\cdots\oplus\Pi_t)
        \leq
        \sum_{i=1}^t\rk(\Pi_i).
    \]
\end{lemma}
\begin{proof}
    Let $\Pi_1,\dots, \Pi_t$ be partitions of $X$ and for $i\in [t]$, let $F_i$ be a forest realizing $\Pi_i$. Let $T$ be a spanning forest of $F_1\cup\dots \cup F_t$. Note that $\Pi_1\oplus\cdots\oplus\Pi_t=\Pi( F_1\cup\dots \cup F_t, X)=\Pi(T, X)$.
    By~\cref{obs:rank_edge_lower_bound}, we have
    $\rk(\Pi_1\oplus\cdots \oplus \Pi_t)=|E(T)|\leq |E(F_1)\cup \dots\cup E(F_t)|\leq \sum_{i=1}^t |E(F_i)|=\sum_{i=1}^t \rk(\Pi_i)$.
\end{proof}

\subsection{ILP Formulation}

\paragraph*{Variables and Bounds.}

We use the following ILP variables:
\begin{equation*}
    s_{v,S}\in\mathbb Z,\qquad 0\le s_{v,S}\le 1
    \qquad(v\in V(G),\ S\in\mathcal S_v).
\end{equation*}
\begin{equation*}
    p^+_{v,\Pi}\in\mathbb Z,\qquad 0\le p^+_{v,\Pi}\le 1
    \qquad(v\in V(G),\ \Pi\in\mathcal P_{v}^{+}).
\end{equation*}
\begin{equation*}
    p^-_{v,\Pi}\in\mathbb Z,\qquad 0\le p^-_{v,\Pi}\le 1
    \qquad(v\in V(G),\ \Pi\in\mathcal P(A_{v}^{-})).
\end{equation*}

Here, for $v \in V(G)$, $s_{v,S} = 1$ will mean that precisely $S \in \mathcal{S}_v$ is selected for the spanning tree,
$p^+_{v,\Pi} = 1$ will mean $\Pi^+_v \in \mathcal P_{v}^{+}$ is the connectivity partition of $A^+_v$ induced by the spanning forest selected for $G_v$,
and finally $p^-_{v,\Pi} = 1$ will mean $\Pi^-_v \in \mathcal P(A_{v}^{-})$ is the connectivity partition of $A^-_v$ induced by the spanning forest selected for $G_v$.

Next, we introduce constraints.
Here, the \emph{choice constraints} will select the edges of the spanning tree and the connectivity partitions w.r.t.\ $A_v^+$ for all $v \in V(G)$. 
Then, together, the \emph{projection, cut, and rank constraints} will enable the inductive \cref{lemma:td_solution_properties}, that is, 
for each $v \in V(G)$, the subgraph selected for $G_v$ is indeed a spanning forest of $G_v$,
the selected subgraph indeed 
induces the chosen connectivity partitions for $A_v^+$ and $A_v^-$,
and finally, every component of the chosen subgraph indeed intersects $A_v^+$. In particular, invoking the lemma for the root of $R$ then gives that the entire chosen subgraph forms a spanning tree of $G$.
Finally, we conclude with the \emph{degree constraints} that encode the degree-requirements for the spanning tree, and the objective, which measures the weight of the selected spanning tree.

\paragraph*{Choice Constraints}

For every $v\in V(G)$, the set of \emph{choice constraints} $\mathcal{C}_v^{\text{choice}}$ is given by
\begin{align}
    \label{choice_s} \sum_{S\in\mathcal S_v}s_{v,S} &= 1 \text{, and}\\
    \label{choice_p_plus}\sum_{\Pi\in\mathcal P_{v}^{+}}p^+_{v,\Pi} &= 1.
\end{align}

\paragraph*{Projection Constraints}

For a finite set $X$, $\Pi\in\mathcal P(X)$, and $Y\subseteq X$, let the \emph{restriction} of $\Pi$ w.r.t.\ $Y$ be defined as 
\[
    \Pi|_Y=\{B\cap Y:B\in\Pi,\ B\cap Y\ne\emptyset\}.
\]
For every $v\in V(G)$,
the set of \emph{projection constraints} $\mathcal{C}_v^{\text{proj}}$ consists of a constraint for every $\Pi\in\mathcal P(A_{v}^{-})$ defined as follows:
\begin{equation}\label{eq:projection}
    p^-_{v,\Pi}
    =
    \sum_{\substack{\Pi'\in\mathcal P_{v}^{+}:\\ \Pi'|_{A_{v}^{-}}=\Pi}}
        p^+_{v,\Pi'}
    .
\end{equation}
The above ensures that 
if $\Pi'$ is the selected connectivity partition for $A_v^+$, then the selected connectivity partition for $A_v^-$ is $\Pi'|_{A_v^-}$.

\paragraph*{Cut Constraints}

We are interested in the connectivity partition induced on $A^+_v$, i.e., we
can safely ``forget'' how vertices of $D^-_v$ are connected to each other,
since no edge outside $G_v$ is incident with a vertex of $D^-_v$.
Thus the only information about $G_v$ that can still matter above $v$ is which
vertices of $A^+_v$ are connected through the selected forest.

Let $c_1,\ldots,c_t$ be the children of $v$.
If $S\in\mathcal S_v$ is selected for $v$ and
$\rho_i\in\mathcal P(A_{c_i}^-)$ is the selected partition for $A_{c_i}^-$,
then $\Pi((A_v^+,S),A_v^+),\rho_1,\ldots,\rho_t$ are all partitions over
$A_v^+$, since $\Pi((A_v^+,S),A_v^+)$ is defined on $A_v^+$ and
$A_{c_i}^-=A_v^+$.
Hence, the graph selected for $G_v$ is the union of the edges $S$ and the
subgraphs selected for the children. On $A_v^+$, the connectivity partition of
this union is represented by the join
\[
    \Pi((A_v^+,S),A_v^+)
    \oplus
    \rho_1
    \oplus\cdots\oplus
    \rho_t.
\]
If $t=0$, i.e., $v$ is a leaf, this join is just $\Pi((A_v^+,S),A_v^+)$.
The cut constraints we are about to state will force whichever partition $\Pi\in\mathcal P_v^+$ is
selected to be finer than this join, i.e.,
\[
    \Pi
    \preceq
    \Pi((A_v^+,S),A_v^+)
    \oplus
    \rho_1
    \oplus\cdots\oplus
    \rho_t.
\]
To encode this, we will use the following lemma, 
which says
it is enough to require that, for
every block $B$ of $\Pi$ and every nonempty proper subset $U\subsetneq B$, at
least one of the partitions in the join touches both sides of the cut
$(U,B\setminus U)$.
\begin{lemma}\label{lemma:cuts_imply_finer_than_join}
    Let $\Pi_1,\ldots,\Pi_t,\Sigma$ be partitions over the same finite set.
    If, for every $B\in \Sigma$ and every nonempty proper subset
    $U\subsetneq B$,
    \[
        1 \leq
        \sum_{i=1}^t\chi(\Pi_i,U,B),
    \]
    then $\Sigma\preceq \Pi_1\oplus\cdots\oplus\Pi_t$.
\end{lemma}
\begin{proof}
    Assume for contradiction that there is a block $B\in \Sigma$ such that no block of $\Pi_1\oplus\cdots \oplus \Pi_t$ contains $B$. In other words, not all elements of $B$ are in the same block of $\Pi_1\oplus\cdots\oplus \Pi_t$. Let $U$ be an inclusion-maximal subset of $B$ that contains elements that are in the same block of $\Pi_1\oplus\cdots\oplus \Pi_t$. By assumption, $U\neq B, \emptyset$. Plugging in $U$ into the above inequality implies that for some $i$, $\chi(\Pi_i, U, B)=1$, i.e. there exists an element of $U$ and an element of $B\setminus U$ that are in the same block of $\Pi_i$, and thus in the same block of $\Pi_1\oplus\cdots\oplus \Pi_t$, which contradicts the maximality of $U$. Therefore, $B$ is fully contained in some block of $\Pi_1\oplus\cdots\oplus \Pi_t$. 
\end{proof}
Formally, for each $v \in V(G)$, let
$
    \mathcal L_v=
    \{(\Pi,B,U): \Pi\in\mathcal P_{v}^{+},\ B\in\Pi,\ \emptyset\ne U\subsetneq B\}.
$
For every $v \in V(G)$,
the set of \emph{cut constraints} $\mathcal{C}_v^{\text{cut}}$
consists of the following constraint for every $(\Pi,B,U)\in\mathcal L_v$:
\begin{equation}\label{eq:cut}
    p^+_{v,\Pi} \leq 
    \sum_{S\in\mathcal S_v}\chi(\Pi((A_v^+,S),A_v^+),U,B)s_{v,S}
    +
    \sum_{\substack{c\in V(G):\\ c\text{ is a child of }v}}
    \sum_{\rho\in\mathcal P(A_{c}^{-})}\chi(\rho,U,B)p^-_{c,\rho}
\end{equation}
If $\Pi$ is not selected, then the left-hand side is zero. If $\Pi$ is
selected, then the constraints impose the condition above for every block
$B\in\Pi$ and every nonempty proper subset $U\subsetneq B$.

\paragraph*{Rank Constraints}

Suppose that $S\in\mathcal S_v$, $\Pi\in\mathcal P_v^+$, and
$\rho_i\in\mathcal P(A_{c_i}^-)$ for $i\in[t]$ are selected.
Recall that the cut constraints give
$
    \Pi
    \preceq
    \Pi((A_v^+,S),A_v^+)
    \oplus
    \rho_1
    \oplus\cdots\oplus
    \rho_t.
$
Hence, by \Cref{obs:rank_monotone}, \[
    \rk(\Pi)
    \leq
    \rk(
        \Pi((A_v^+,S),A_v^+)
        \oplus
        \rho_1
        \oplus\cdots\oplus
        \rho_t
    ).\]
On the other hand, by \Cref{obs:rank_subadditive_join}, we have
\[
    \rk(
        \Pi((A_v^+,S),A_v^+)
        \oplus
        \rho_1
        \oplus\cdots\oplus
        \rho_t
    )
    \leq
    \rk(\Pi((A_v^+,S),A_v^+))
    +
    \sum_{i=1}^t \rk(\rho_i).
\]
The rank constraints we give next say that the right-hand side of the above
inequality shall equal $\rk(\Pi)$. Together with the two inequalities above,
this forces $\rk(\Pi)$, the rank of the join, and the sum of the ranks of
$\Pi((A_v^+,S),A_v^+)$ and $\rho_1,\ldots,\rho_t$ to be equal. Thus, 
the set of \emph{rank constraints} $\mathcal{C}_v^\text{rank}$ is given by the constraint:
\begin{equation}
    \sum_{S\in\mathcal S_v}\rk(\Pi((A_v^+,S),A_v^+))s_{v,S}
    +
    \sum_{\substack{c\in V(G):\\ c\text{ is a child of }v}}
    \sum_{\rho\in\mathcal P(A_{c}^{-})}\rk(\rho)p^-_{c,\rho}
    =
    \sum_{\Pi\in\mathcal P_{v}^{+}}\rk(\Pi)p^+_{v,\Pi}.
\end{equation}

At this point, we know that the selected partition $\Pi$ is finer than the join
above and has the same rank as this join. To conclude that the two partitions
are equal, we will use the following elementary fact.

\begin{lemma}\label{lemma:partitions_of_same_rank}
    Let $\Pi_1, \Pi_2$ be two partitions with $\Pi_1 \preceq \Pi_2$ and $\rk(\Pi_1) = \rk(\Pi_2)$. Then $\Pi_1 = \Pi_2$.
\end{lemma}
\begin{proof}
From $\Pi_1 \preceq \Pi_2$, it follows that for each block $B \in \Pi_1$,
there is a block $f(B) \in \Pi_2$ with $B \subseteq f(B)$, where $f : \Pi_1 \to \Pi_2$.
We aim to show that $f$ is the identity function.
By $\rk(\Pi_1) = \rk(\Pi_2)$, we have $|\Pi_1| = |\Pi_2|$.
We next show that $f$ is surjective.
Let $B' \in \Pi_2$, and let $b \in B'$.
Further, let $B \in \Pi_1$ such that $b \in B$. Then $f(B) = B'$, since we have $ \{ b \} \subseteq B \subseteq f(B)$ and $B'$ is the only block of $\Pi_2$ that contains $b$.
Thus $f$ is bijective.

Towards a contradiction, assume $f$ is not the identity,
that is, there is $B \in \Pi_1$
with $B \neq f(B)$, i.e., by $B \subseteq f(B)$, this means there is $b \in f(B) \setminus B$.
Further, let $C \in \Pi_1$ such that $b \in C$.
Then $f(C)=f(B)$, since both are blocks of $\Pi_2$ that contain $b$.
Applying $f^{-1}$ yields $C=B$.
But $b \in C$ and $b \not\in B$, a contradiction.
\end{proof}

Apart from the above, when deriving \cref{lemma:td_solution_properties}, it will remain to show that the subgraph selected for $G_v$ is
acyclic. Recall from
\Cref{obs:rank_edge_lower_bound} that a forest realizing a partition $\Sigma$
has exactly $\rk(\Sigma)$ edges. Thus, intuitively, once the rank constraints
give that the rank of the join equals the sum of the ranks of the partitions
being joined, we can take forests realizing $\Pi((A_v^+,S),A_v^+)$ and
$\rho_1,\ldots,\rho_t$, form their union, and obtain a forest realizing the join
of the partitions, i.e., the connectivity partition w.r.t.\ $A_v^+$ of the
subgraph selected for $G_v$.
We formalize this with
the following lemma.
\begin{lemma}\label{lemma:forest_union_rank}
    Let $H_1,\ldots,H_t$ be pairwise edge-disjoint forests, and let $X$ be a
    finite set of vertices such that $V(H_i)\cap V(H_j)\subseteq X$ for all
    distinct $i,j\in[t]$.
    Then $H_1\cup\cdots\cup H_t$ is a forest if and only if
    \[
        \sum_{i=1}^t\rk(\Pi(H_i,X))
        =
        \rk(\Pi(H_1,X)\oplus\cdots\oplus\Pi(H_t,X)).
    \]
\end{lemma}

\begin{proof}
    For each $i\in[t]$, define a graph $H_i'$ on vertex set $X$ as follows:
    for every connected component $C$ of $H_i$ with $C\cap X\neq\emptyset$,
    choose an arbitrary tree on $C\cap X$, and let $H_i'$ be the union of the
    trees. Let $H'$ be the multigraph obtained as the edge-disjoint union of
    $H_1',\ldots,H_t'$ on the common vertex set $X$.
    Then $|E(H')|=\sum_{i=1}^t\rk(\Pi(H_i,X))$ and
    $\Pi(H',X)=\Pi(H_1,X)\oplus\cdots\oplus\Pi(H_t,X)$.
    Hence the displayed equality is equivalent to
    $|E(H')|=\rk(\Pi(H',X))$, which holds if and only if $H'$ is acyclic (where two parallel edges are understood to form a cycle) by
    applying \Cref{obs:rank_edge_lower_bound} to a spanning forest of $H'$.
    Finally, observe $H'$ is acyclic if
    and only if $H_1\cup\cdots\cup H_t$ is acyclic, hence the statement follows.
\end{proof}

\paragraph*{Degree Constraints}

\begin{observation}\label{obs:owned_incident_edges}
    Let $S_u\in\mathcal S_u$ for every $u\in V(G)$, i.e., 
    a selection of edges $\bigcup_{u\in V(G)}S_u \subseteq E(G)$. 
    Then for every $v\in V(G)$,
    the selected edges incident with $v$ are
    \[
        S_v
        \;\dot\cup\;
        \{vw:w\in D_v^-\text{ and }vw\in S_w\}.
    \]
\end{observation}

Motivated by \Cref{obs:owned_incident_edges}, for every $v\in V(G)$, the set of
\emph{degree constraints} $\mathcal{C}_v^\text{deg}$ is given by the constraint:
\begin{equation}\label{eq:degree}
    \sum_{S\in\mathcal S_v}|S|s_{v,S}
    +
    \sum_{w\in D_{v}^{-}}
    \sum_{\substack{S\in\mathcal S_w:\\ vw\in S}}s_{w,S}
    \leq
    \req(v).
\end{equation}

\paragraph*{Objective}

Since the sets $O(v)$ partition $E(G)$, the following objective sums the weights
of the selected edges exactly once:
\[
    \min
    \sum_{v\in V(G)}
    \sum_{S\in\mathcal S_v}
        \left(\sum_{e\in S}\weightf(e)\right)s_{v,S}.
\]

\subsection{Proof of Correctness}

We call a mapping $\alpha$ from the ILP variables to $\mathbb{Z}$ a
\emph{feasible solution to the ILP} if each variable is mapped to a value of the
variable's domain and all ILP constraints are satisfied. If additionally the objective value for $\alpha$ is minimum among all feasible solutions, $\alpha$ is called an \emph{optimal solution to the ILP}.

Let $\alpha$ be a feasible solution to the ILP.
For $v\in V(G)$, let $S_v^\alpha$ be the unique set $S\in\mathcal S_v$ with
$\alpha(s_{v,S})=1$.
We write $F_v^\alpha$ for the subgraph selected for $G_v$, i.e.,
\[
    F_v^\alpha
    =
    \left(V(G_v),\bigcup_{u\in D_v^+}S_u^\alpha\right).
\]
Furthermore, let $\Pi_v^{+,\alpha}\in\mathcal P_v^+$ be the unique partition
with $\alpha(p^+_{v,\Pi_v^{+,\alpha}})=1$, and let
$\Pi_v^{-,\alpha}\in\mathcal P(A_v^-)$ be the unique partition with
$\alpha(p^-_{v,\Pi_v^{-,\alpha}})=1$.
By the choice and projection constraints, these objects are well-defined. Our first goal is to show that $F_v^\alpha$ is a spanning forest that induces the chosen connectivity partition on $A_v^+$ and $A_v^-$. Before that, we will need an auxiliary lemma.

\begin{lemma}\label{lemma:partition_union_join}
    Let $H_1,\ldots,H_t$ be graphs, and let $X$ be a finite set of vertices
    such that $V(H_i)\cap V(H_j)\subseteq X$ for all distinct
    $i,j\in[t]$.
    Then
    \[
        \Pi(H_1\cup\cdots\cup H_t,X)
        =
        \Pi(H_1,X)\oplus\cdots\oplus\Pi(H_t,X).
    \]
\end{lemma}
\begin{proof}
         For $i\in [t]$, let $F_i$ be a spanning forest of $H_i$. Then $\Pi(H_1,X)\oplus \cdots \oplus\Pi(H_t, X)=\Pi(F_1, X)\oplus\cdots\oplus \Pi(F_t, X)=\Pi(F_1\cup\dots\cup F_t, X)$. Since $F_1\cup\cdots\cup F_t$ is a subgraph of $H_1\cup\dots \cup H_t$, it suffices to show that if two vertices are in the same connected component of $H_1\cup\dots\cup H_t$, they are in the same component of $F_1\cup \dots \cup F_t$. Let $u,v\in V(H_1\cup \dots \cup H_t)$ be in the same connected component of $H_1\cup \dots \cup H_t$ and let $\pi$ be an arbitrary path between them in $H_1\cup\dots \cup H_t$.
    We will split the path $\pi$ into subpaths that are fully contained in some $H_i$. More precisely, let $u=u_1, u_2,\dots, u_\ell=v$ be vertices of $\pi$ such that for all $i\in [\ell-1]$ there is a $j\in [t]$ such that the subpath of $\pi$ from $u_i$ to $u_{i+1}$ is contained in $H_j$. As $u_i$ and $u_{i+1}$ are connected in $H_j$, they are connected in $F_j$. Thus $u$ and $v$ are in the same connected component of $F_1\cup\dots\cup F_t$. 
\end{proof}

\begin{lemma}\label{lemma:td_solution_properties}
Let $\alpha$ be a feasible solution to the ILP.
Then, for each $v \in V(G)$, $F_v^\alpha$ is a spanning forest of $G_v$ with
$\Pi(F_v^\alpha, A_v^+) = \Pi_v^{+,\alpha}$,
$\Pi(F_v^\alpha, A_v^-) = \Pi_v^{-,\alpha}$,
and every connected component of $F_v^\alpha$ contains at least one vertex of $A^+_v$.
\end{lemma}
\begin{proof}
We proceed by induction over the elimination tree $R$.

\emph{Base Case.}
Let $v$ be a leaf of $R$.
Observe that $D_v^+=\{v\}$, and hence $F_v^\alpha=(V(G_v),S_v^\alpha)$.
Since $S_v^\alpha\subseteq O(v)$, every edge of $F_v^\alpha$ is incident with $v$.
Thus $F_v^\alpha$ is a forest, and since it has vertex set $V(G_v)$, it is a spanning forest of $G_v$.
By the rank constraints, we have
$\rk(\Pi((A_v^+,S_v^\alpha),A_v^+)) = \rk(\Pi_v^{+,\alpha})$.
Since $F_v^\alpha$ has the same edge set as $(A_v^+,S_v^\alpha)$, by the
definition of $\Pi(\cdot,\cdot)$, we have
$\Pi(F_v^\alpha,A_v^+) = \Pi((A_v^+,S_v^\alpha),A_v^+)$.
Thus $\rk(\Pi(F_v^\alpha,A_v^+)) =  \rk(\Pi_v^{+,\alpha})$.
Since $v$ is a leaf, the cut constraints and
\Cref{lemma:cuts_imply_finer_than_join} yield
$\Pi_v^{+,\alpha}\preceq \Pi((A_v^+,S_v^\alpha),A_v^+)=\Pi(F_v^\alpha,A_v^+)$.
Applying \Cref{lemma:partitions_of_same_rank} yields
$\Pi(F_v^\alpha, A_v^+) = \Pi_v^{+,\alpha}$.
By the projection constraints, we have
$\Pi_v^{-,\alpha}=\Pi_v^{+,\alpha}|_{A_v^-}$.
Moreover, by the definition of $\Pi(\cdot,\cdot)$, we have
$\Pi(F_v^\alpha,A_v^-)=\Pi(F_v^\alpha,A_v^+)|_{A_v^-}$.
Thus $\Pi(F_v^\alpha, A_v^-) = \Pi_v^{-,\alpha}$.
Finally, since $G_v$ is the edge-induced subgraph of $G$ with edge set $O(v)$,
we have $V(G_v)\subseteq A_v^+$.
Hence every connected component of $F_v^\alpha$ contains at least one vertex of $A_v^+$.

\emph{Inductive Case.}
Let $v$ be a non-leaf of $R$, and let $c_1,\ldots,c_t$ be the children of $v$.
Assume the statement holds for $c_1,\ldots,c_t$.
Since $A_{c_i}^-=A_v^+$, the induction hypothesis gives
$\Pi_{c_i}^{-,\alpha}=\Pi(F_{c_i}^\alpha,A_v^+)$ for each
$i\in[t]$.
Moreover, the pairwise intersections of the vertex sets of
$(A_v^+,S_v^\alpha),F_{c_1}^\alpha,\ldots,F_{c_t}^\alpha$ are contained in
$A_v^+$.
As the graph
$(A_v^+,S_v^\alpha)\cup F_{c_1}^\alpha\cup\cdots\cup F_{c_t}^\alpha$
has the same edge set as $F_v^\alpha$, by \Cref{lemma:partition_union_join},
\[
    \Pi(F_v^\alpha,A_v^+)
    =
    \Pi((A_v^+,S_v^\alpha),A_v^+)
    \oplus
    \Pi_{c_1}^{-,\alpha}
    \oplus\cdots\oplus
    \Pi_{c_t}^{-,\alpha}.
\]
By the rank constraints, we have
\[
    \rk(\Pi_v^{+,\alpha})
    =
    \rk(\Pi((A_v^+,S_v^\alpha),A_v^+))
    +
    \sum_{i=1}^t \rk(\Pi_{c_i}^{-,\alpha}).
\]
By the cut constraints and \Cref{lemma:cuts_imply_finer_than_join}, we have
\[
    \Pi_v^{+,\alpha}
    \preceq
    \Pi((A_v^+,S_v^\alpha),A_v^+)
    \oplus
    \Pi_{c_1}^{-,\alpha}
    \oplus\cdots\oplus
    \Pi_{c_t}^{-,\alpha}.
\]
	By \Cref{obs:rank_subadditive_join}, we get
	\[
	    \begin{aligned}
	    &\rk(
	        \Pi((A_v^+,S_v^\alpha),A_v^+)
	        \oplus
	        \Pi_{c_1}^{-,\alpha}
	        \oplus\cdots\oplus
	        \Pi_{c_t}^{-,\alpha}
	    )
	    \\
	    &\qquad
	    \leq
	    \rk(\Pi((A_v^+,S_v^\alpha),A_v^+))
	    +
	    \sum_{i=1}^t\rk(\Pi_{c_i}^{-,\alpha})
	    =
	    \rk(\Pi_v^{+,\alpha}).
	    \end{aligned}
	\]
Since
$\Pi_v^{+,\alpha}
\preceq
\Pi((A_v^+,S_v^\alpha),A_v^+)
\oplus
\Pi_{c_1}^{-,\alpha}
\oplus\cdots\oplus
\Pi_{c_t}^{-,\alpha}$,
by \Cref{obs:rank_monotone}, the rank of the partition on the right is at least
$\rk(\Pi_v^{+,\alpha})$.
Combining this with the previous inequality gives
\[
    \rk(\Pi_v^{+,\alpha})
    =
    \rk(
        \Pi((A_v^+,S_v^\alpha),A_v^+)
        \oplus
        \Pi_{c_1}^{-,\alpha}
        \oplus\cdots\oplus
        \Pi_{c_t}^{-,\alpha}
    ).
\]
Applying \Cref{lemma:partitions_of_same_rank} yields
\[
    \Pi(F_v^\alpha,A_v^+)
    =
    \Pi_v^{+,\alpha}.
\]
Observe that
$(A_v^+,S_v^\alpha),F_{c_1}^\alpha,\ldots,F_{c_t}^\alpha$ are pairwise
edge-disjoint forests.
Thus, by the previous rank equality, \Cref{lemma:forest_union_rank} gives that 
$(A_v^+,S_v^\alpha)\cup F_{c_1}^\alpha\cup\cdots\cup F_{c_t}^\alpha$ is a
forest.
Since $F_v^\alpha$ has the same edge set as this forest and has vertex-set
$V(G_v)$, it follows that $F_v^\alpha$ is a spanning forest of $G_v$.

By the projection constraints, we have
$\Pi_v^{-,\alpha}=\Pi_v^{+,\alpha}|_{A_v^-}$.
Moreover, by the definition of $\Pi(\cdot,\cdot)$, we have
$\Pi(F_v^\alpha,A_v^-)=\Pi(F_v^\alpha,A_v^+)|_{A_v^-}$.
Thus $\Pi(F_v^\alpha, A_v^-) = \Pi_v^{-,\alpha}$.

Finally, as in the base case, every connected component of $(A_v^+,S_v^\alpha)$ contains a vertex of
$A_v^+$.
For each $i\in[t]$, every connected component of $F_{c_i}^\alpha$
contains a vertex of $A_v^+$ as well, as by induction it contains a vertex of
$A_{c_i}^+$, and if this vertex is $c_i$, then
$\{c_i\}\notin\Pi_{c_i}^{+,\alpha}$ implies that the component containing
$c_i$ also contains a vertex of $A_{c_i}^-=A_v^+$.
Therefore every connected component of $F_v^\alpha$ contains at least one
vertex of $A_v^+$.
\end{proof}

Next, we will show how to construct a spanning tree satisfying the degree constraints from a feasible solution to the ILP, and vice versa.

\begin{lemma}\label{lemma:td_soundness}
    Let $\alpha$ be a feasible solution to the ILP. Then, there exists a spanning tree $T$ of $G$ respecting the degree requirement $\req$ such that
    $T$'s weight equals the objective value of $\alpha$.
\end{lemma}
\begin{proof}
Let $T=F_r^\alpha$.
Note that $G_r=G$.
By \Cref{lemma:td_solution_properties}, $T$ is a spanning forest of $G$ and
every connected component of $T$ contains a vertex of $A_r^+=\{r\}$.
Thus $T$ is connected, and hence a spanning tree of $G$.

For $v\in V(G)$, the degree constraint in $\mathcal C_v^\text{deg}$ 
under $\alpha$ gives
\[
    |S_v^\alpha|
    +
    \sum_{w\in D_v^-}|\{vw\}\cap S_w^\alpha|
    \leq
    \req(v).
\]
By the definition of $F_r^\alpha$, we have
$E(T)=\bigcup_{u\in V(G)}S_u^\alpha$.
Thus, by \Cref{obs:owned_incident_edges}, the left-hand side is the number of
edges of $T$ incident with $v$. Hence $\deg_T(v)\leq \req(v)$, so $T$ respects $\req$.

Finally, since $E(T)=\bigcup_{v\in V(G)}S_v^\alpha$ and the sets $O(v)$
partition $E(G)$, the objective value of $\alpha$ is
\[
    \sum_{v\in V(G)}\sum_{e\in S_v^\alpha}\weightf(e)
    =
    \sum_{e\in E(T)}\weightf(e).\qedhere
\]
\end{proof}

\begin{lemma}\label{lemma:join_implies_cuts} 
    Let $\Pi_1,\ldots,\Pi_t,\Sigma$ be partitions over the same finite set.
    If
    \[
        \Sigma=\Pi_1\oplus\cdots\oplus\Pi_t,
    \]
    then, for every $B\in\Sigma$ and every nonempty proper subset
    $U\subsetneq B$,
    \[
        1\leq \sum_{i=1}^t \chi(\Pi_i,U,B).
    \]
\end{lemma}
\begin{proof}
    For $i\in [t]$, let $F_i$ be a forest realizing $\Pi_i$.
    Suppose for contradiction that for some $B\in \Sigma$ and $\emptyset\neq U\subsetneq B$, we have that $\chi(\Pi_i, U, B)=0$ for all $i\in [t]$. Let $u\in U$ and $v\in B\setminus U$. As $u$ and $v$ are in the same block of $\Sigma$, there is a $uv$-path $\pi$ in $F_1\cup\dots F_t$. Consider an edge $u'v'$ on $\pi$ such that $u'\in U$ and $v'\in B\setminus U$. Let $i$ be such that $u'v'\in E(F_i)$. In other words, $u'\in U$ and $v'\in B\setminus U$ are in the same block of $\Pi_i$, which contradicts the assumption that $\chi(\Pi_i, U, B)=0$.  
\end{proof}

\begin{lemma}\label{lemma:td_completeness}
    Let $T$ be a spanning tree respecting the degree requirement $\req$.
    Then there is a feasible solution $\alpha$ to the ILP whose objective value equals the weight of $T$.
\end{lemma}
\begin{proof}
We construct an integer assignment $\alpha$ to the ILP variables as follows.
For every $v\in V(G)$, set
\[
    S_v \coloneqq E(T)\cap O(v),\qquad
    F_v \coloneqq (V(G_v),E(T)\cap E_v),
\]
and
\[
    \Pi_v^+ \coloneqq \Pi(F_v,A_v^+),
    \qquad
    \Pi_v^- \coloneqq \Pi(F_v,A_v^-).
\]
We have $\Pi_r^+=\{\{r\}\}\in\mathcal P_r^+$, and $\Pi_v^+\in\mathcal P_v^+$
for every $v\neq r$, as otherwise the block of $\Pi_v^+$ containing $v$ would
be $\{v\}$, and the component of $F_v$ containing $v$ would be a connected
component of $T$ not containing $r$, contradicting that $T$ is connected.
Hence it is well-defined, for each $v\in V(G)$, to set
$\alpha(s_{v,S_v})=\alpha(p^+_{v,\Pi_v^+})=\alpha(p^-_{v,\Pi_v^-})=1$,
and to set all remaining variables for $v$ to 0.

Next, we argue that $\alpha$ is a feasible solution to the ILP.
The choice constraints hold by construction. For every $v\in V(G)$, by the
definition of $\Pi(\cdot,\cdot)$,
\[
    \Pi_v^-=\Pi(F_v,A_v^-)=\Pi(F_v,A_v^+)|_{A_v^-}=\Pi_v^+|_{A_v^-},
\]
hence the projection constraints hold.

Fix $v\in V(G)$, and let $c_1,\ldots,c_t$ be the children of $v$ (note that
$t=0$ if $v$ is a leaf). For every $i\in[t]$, we have
$A_{c_i}^-=A_v^+$ and hence $\Pi_{c_i}^-=\Pi(F_{c_i},A_v^+)$. Moreover, $F_v$
has the same edge set as
\[
    (A_v^+,S_v)\cup F_{c_1}\cup\cdots\cup F_{c_t}.
\]
The graphs $(A_v^+,S_v),F_{c_1},\ldots,F_{c_t}$ are forests, have pairwise
disjoint edge sets, and the pairwise intersections of their vertex sets are
contained in $A_v^+$.
Hence, by \Cref{lemma:partition_union_join},
\[
    \Pi_v^+
    =
    \Pi((A_v^+,S_v),A_v^+)
    \oplus
    \Pi_{c_1}^-
    \oplus\cdots\oplus
    \Pi_{c_t}^-.
\]
Thus \Cref{lemma:join_implies_cuts} gives the cut constraint indexed by every
triple $(\Pi,B,U)\in\mathcal L_v$ with $\Pi=\Pi_v^+$. For every
$(\Pi,B,U)\in\mathcal L_v$ with $\Pi\neq \Pi_v^+$, we have
$\alpha(p^+_{v,\Pi})=0$, hence the constraint holds trivially.

Since the edge set of $(A_v^+,S_v)\cup F_{c_1}\cup\cdots\cup F_{c_t}$ is
contained in $E(T)$, the union is a forest, and
applying \Cref{lemma:forest_union_rank} yields
\[
    \rk(\Pi((A_v^+,S_v),A_v^+))
    +
    \sum_{i=1}^t\rk(\Pi_{c_i}^-)
    =
    \rk(\Pi_v^+),
\]
which is the constraint in $\mathcal C_v^\text{rank}$ under $\alpha$.

Finally, by \Cref{obs:owned_incident_edges} and
$\bigcup_{u\in V(G)}S_u=E(T)$, the left-hand side of the degree constraint in
$\mathcal C_v^\text{deg}$ is the number of edges of $T$ incident with $v$, and
hence equals $\deg_T(v)$. Since $T$ respects $\req$, the degree constraint for
$v$ holds.

It remains to show that the objective value of $\alpha$ equals the weight of $T$.
Since $S_v=E(T)\cap O(v)$ and the sets $O(v)$ partition $E(G)$, the objective
value of $\alpha$ is
\[
    \sum_{v\in V(G)}\sum_{e\in S_v}\weightf(e)
    =
    \sum_{e\in E(T)}\weightf(e),
\]
as required.
\end{proof}

Now we are ready to prove the main result of this section.

\thmtd*
\begin{proof}
Since all variables of the ILP have bounded domains, if the ILP is feasible,
then its objective is bounded from below, and hence the ILP has an optimal
solution. Thus, by \Cref{lemma:td_soundness,lemma:td_completeness}, solving the
ILP yields a minimum-weight spanning tree satisfying the degree requirements, if
one exists.

It thus remains to construct and then solve the ILP.
First, we compute an optimal elimination tree in \FPT-time parameterized by
$\td(G)$, for example using the algorithm of Nadara, Pilipczuk, and
Smulewicz~\cite{DBLP:conf/esa/NadaraPS22}. With this at hand, the ILP can clearly be constructed in \FPT-time
parameterized by $\td(G)$.
We aim to solve the constructed ILP via Theorem 6 of \cite{treedepthilp}. Towards this, we first transform it into the following normal form:
\begin{equation*}
    \min\{wx: Ax=b,\ l\le x\le u,\ x\in\mathbb Z^n\},
\end{equation*}
where $A\in\mathbb Z^{m\times n}$, $b\in\mathbb Z^m$,
$w\in\mathbb Z^n$, and $l,u\in(\mathbb Z\cup\{\pm\infty\})^n$.
Equivalently, each constraint shall be of the form $L = R$, where $L$ is a linear combination of ILP variables, and $R$ is a constant. Note that we did not adhere to this normal form purely for ease of presentation.

The theorem states that an optimal solution to such an ILP can be computed in \FPT-time  and the maximum constraint coefficient, i.e., 
$\|A\|_\infty \coloneqq \max_{i,j} |A_{ij}|$.

\subparagraph*{Normalizing the ILP.}
The choice constraints are already in the desired form.
The projection constraints can be rearranged so that the right-hand side equals
zero, and the same is true for the rank constraints.
For each constraint in $\mathcal C_v^{\text{cut}}$, indexed by
$(\Pi,B,U)\in\mathcal L_v$, let $L_{v,\Pi,B,U}$ and $R_{v,\Pi,B,U}$
denote the left-hand side and right-hand side of \Cref{eq:cut}.
We first rearrange the constraint as
$0\leq R_{v,\Pi,B,U}-L_{v,\Pi,B,U}$, introduce a \emph{slack variable}
\[
    \lambda_{v,\Pi,B,U}\in\mathbb Z,\qquad
    0\leq \lambda_{v,\Pi,B,U}\leq +\infty,
\]
and replace the inequality by
$\lambda_{v,\Pi,B,U}=R_{v,\Pi,B,U}-L_{v,\Pi,B,U}$.
For the constraint in $\mathcal C_v^\text{deg}$, let $L_v$ denote the
left-hand side of \Cref{eq:degree}. We again introduce a slack variable
\[
    \mu_v\in\mathbb Z,\qquad 0\leq \mu_v\leq +\infty,
\]
and replace the degree constraint by $L_v+\mu_v=\req(v)$.
In total, this yields an equivalent ILP of the desired form.
In a slight abuse of notation, in the rest of this proof, when we refer to ``the ILP'' or its constraints, we are referring to this normal form, not the original definition.

\subparagraph{Bounding $\|A\|_\infty$.}
All constraint coefficients  (i.e., the entries of $A$) used in the choice, projection, and cut constraints lie in $\{-1,0,1\}$.
The coefficients in the rank constraints arise from expressions $\rk(\Pi)$ where $\Pi$ is a partition over either $A_v^+$ or $A_v^-$. 
Since both sets have at most $k$ elements, by the definition of $\rk(\cdot)$,
$\rk(\Pi) \leq k$.
Finally, in the degree constraints the only coefficients that are not zero or one are the
coefficients $|S|$ of variables $s_{v,S}$, which satisfy
$|S|\leq |O(v)|\leq k$.
Hence $\|A\|_\infty \leq k$.\footnote{Note that the degree requirements $\req(v)$ appear only as entries of $b$ and not as coefficients in~$A$.}

\subparagraph{Bounding $\td(G_D(A))$.}

To bound the treedepth of the dual graph, we first count the constraints
associated with each $v\in V(G)$.
Let $B_i$ denote the $i$-th \emph{Bell number}, that is, the number of
partitions of an $i$-element set.
Since $|A_v^+|\leq k$, $|A_v^-|\leq k$, and $|O(v)|\leq k$, we have
\begin{itemize}
    \item $|\mathcal{C}_v^{\text{choice}}| = 2$,
    \item $|\mathcal{C}_v^{\text{proj}}| = B_{|A_v^-|}\leq B_k$,
    \item $|\mathcal{C}_v^{\text{cut}}| = |\mathcal{L}_v|\leq B_k\cdot k\cdot 2^k$,
    \item $|\mathcal{C}_v^\text{rank}| = 1$, and
    \item $|\mathcal{C}_v^\text{deg}| = 1$.
\end{itemize}
Let $\mathcal{C}_v \coloneqq \mathcal{C}_v^{\text{choice}} \;\dot\cup\; \mathcal{C}_v^{\text{proj}} \;\dot\cup\; \mathcal{C}_v^{\text{cut}} \;\dot\cup\; \mathcal{C}_v^\text{rank} \;\dot\cup\; \mathcal{C}_v^\text{deg}$,
and let $\mathcal{C} \coloneqq \dot\bigcup_{v\in V(G)} \mathcal{C}_v$ be the set of all ILP constraints.
For a constraint $c \in \mathcal{C}$, let $\textup{var}(c)$ denote the set of ILP variables occurring with non-zero coefficient in $c$.

For an ILP variable $x$, we write $\nu(x)$ for the unique vertex $v \in V(G)$ associated with $x$,
that is 
\begin{align*}
    \nu(s_{v,S}) &\coloneqq v, \\
    \nu(p^+_{v,\Pi}) &\coloneqq v, \\
    \nu(p^-_{v,\Pi}) &\coloneqq v, \\
    \nu(\lambda_{v,\Pi,B,U}) &\coloneqq v \text{, and} \\
    \nu(\mu_v) &\coloneqq v.
\end{align*}
We further extend the notation $\nu$ to constraints $c \in \mathcal{C}$:
\begin{equation*}
    \nu(c) \coloneqq \{\nu(x): x\in\textup{var}(c)\}.
\end{equation*}

Observe that for $c_1, c_2 \in \mathcal{C}$, if
$\textup{var}(c_1) \cap \textup{var}(c_2) \neq \emptyset$, then
$\nu(c_1)\cap\nu(c_2)\neq\emptyset$.
Consequently, the graph on vertex set $\mathcal C$ in which two distinct
constraints $c_1,c_2$ are adjacent whenever
$\nu(c_1)\cap\nu(c_2)\neq\emptyset$, call it $G^\nu_D(A)$, is a supergraph of $G_D(A)$.
It hence suffices to bound the treedepth of $G^\nu_D(A)$.

Towards this, the central fact guaranteed by our construction is the following: for any $c \in \mathcal{C}_v$, we have $\nu(c) \subseteq D^+_v$, i.e., the only vertices associated with
variables occurring in a constraint $c\in\mathcal C_v$ are $v$ and descendants
of $v$.
We now construct an elimination tree $R^\nu$ for $G^\nu_D(A)$.
Starting with $R$, replace each vertex $v$ by a path $P_v$ with vertex set $\mathcal C_v$ in arbitrary order, by attaching the parent of $v$ to the first path vertex, and the children of $v$ to the last path vertex.
For each $v\in V(G)$, we have $|\mathcal C_v|\leq B_k\cdot k\cdot 2^k+B_k+4$.
Hence, the height of $R^\nu$ is bounded by $k (B_k\cdot k\cdot 2^k+B_k+4)$.
It remains to show that $R^\nu$ is an elimination tree of $G^\nu_D(A)$.
Let $c_1 c_2$ be an edge of $G^\nu_D(A)$.
We need to show that $c_1$ is an ancestor or descendant of $c_2$ in $R^\nu$.
Let $v,w\in V(G)$ be such that $c_1\in\mathcal C_v$ and
$c_2\in\mathcal C_w$.
Since $c_1c_2$ is an edge of $G^\nu_D(A)$, there is a vertex
$z\in\nu(c_1)\cap\nu(c_2)$.
By the central fact above, $z\in D_v^+\cap D_w^+$.
Thus $v$ and $w$ both lie on the root--$z$ path in $R$.
After replacing vertices of $R$ by paths, all vertices of $P_v$ and $P_w$ still
lie on one root--leaf path in $R^\nu$.
Hence, in particular, $c_1$ is an ancestor or descendant of $c_2$ in $R^\nu$.

Therefore $\td(G_D(A))\leq \td(G^\nu_D(A))\leq k (B_k\cdot k\cdot 2^k+B_k+4) $. Hence Theorem~6 of \cite{treedepthilp} is applicable and the statement follows. 
\end{proof}

 \section{W-Hardness for \setmstp}
\label{section:fvs_w1}

In this section, we show that \setmstp is
$W[1]$-hard parameterized by the feedback vertex number, that is, the minimum number of vertices to delete in a graph to make it acyclic.
We reduce from the strongly $W[1]$-hard \textsc{Simple Multidimensional Partitioned Subset Sum (SMPSS)} problem \cite{GanianOR23}.
In this problem, we are given a \emph{dimension} $d \in \mathbb{N}$, $\ell \in \mathbb{N}$ sets of vectors $P_1, \dots, P_\ell$ in $\mathbb{N}_0^d$, and a \emph{target vector} $\bar{t} \in \mathbb{N}_0^d$.
For each $i \in [\ell]$, we write $P_i = \{ p_i^1, \dots, p_i^{|P_i|} \}$. Furthermore, we may assume each $p_i^j$ has exactly one non-zero component in dimension $d_i^j$.
For $v \in \mathbb{N}^d$, we write $\|v\|$ for the $L_1$-norm of $v$, that is, $\|v\| = \sum_{k \in [d]} |v_k|$.
In particular, $(p_i^j)_{d_i^j} = \| p_i^j \|$.
	The parameter of the problem is $d$ and the objective is to decide whether there are $\bar{p}_1 \in P_1, \dots, \bar{p}_\ell \in P_\ell$ such that $\sum_{i\in[\ell]} \bar{p}_i = \bar{t}$.
Note that we can assume that $\bar t_k\leq \sum_{i\in[\ell]}\sum_{j\in[|P_i|]}(p_i^j)_k$ for every $k\in[d]$, as otherwise the \textsc{SMPSS} instance is trivially negative.

\subparagraph{Computing the reduction.}
We construct the unweighted instance $(G, \req)$ of \setmstp as follows.
For each $i \in [\ell]$, we define a \emph{selection gadget} that models the choice of $\bar{p}_i \in P_i$.
We start with the complete bipartite graph $K_{1,|P_i|}$
with vertex set $\{ s_i \} \dot\cup \{ {s'}_i^j \mid j \in [|P_i|]\}$.
Here, we call $s_i$ the $i$'th \emph{selection vertex}, and $\{ {s'}_i^1, \dots, {s'}_i^{|P_i|} \}$ the $i$'th \emph{input vertices}.
We set $\req(s_i) = \{1\}$ and $\req({s'}_i^j) = \{1,3\}$ for all $j \in [|P_i|]$.
To each of the $i$'th input vertices, we attach a component as follows.
Let $j \in [|P_i|]$. We make adjacent to ${s'}_i^j$ the \emph{bridge vertex} $b_i^j$, and make adjacent to $b_i^j$ the \emph{gate vertex} $g_i^j$.
Further, the gate vertex is adjacent to $\|p_i^j\|$ \emph{output vertices}, denoted by $O_i^j$.
We set $\req(b_i^j) = \{1\}$, $\req(g_i^j) = \{1, \|p_i^j\| + 2\}$, and $\req(o) = \{1\}$ for all $o \in O_i^j$.

Next, we model the target vector. For each dimension $k \in [d]$, we create a \emph{target vertex} $t_k$ and set $\req(t_k) = \{ \bar{t}_k + 1 \} $.
For each $i \in [\ell], j \in [|P_i|]$, we create all possible edges between the output vertices $O_i^j$ and the target vertex $t_{d_i^j}$.

Finally, we add the \emph{root vertex} $r$ and make it adjacent to all input, gate, and target vertices.
We have $\deg_G(r)=d + 2 \sum_{i \in [\ell]} |P_i|$ and set $\req(r) = \{\deg_G(r)\}$.

\begin{figure}
    \centering
    \includegraphics[page=2]{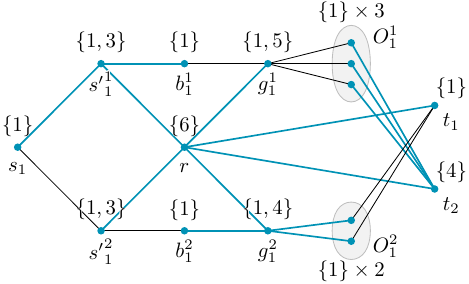}
    \caption{An example of the reduction with $\ell = 1$, $P_1 = \left\{ \begin{pmatrix} 0 \\ 3 \end{pmatrix}, \begin{pmatrix} 2 \\ 0 \end{pmatrix} \right\}$, and $\bar{t} = \begin{pmatrix} 0 \\ 3 \end{pmatrix}$. The set accompanying each vertex $v$ is $\req(v)$. The solution to the instance is $\bar{p}_1 = p_1^1 = \begin{pmatrix} 0 \\ 3 \end{pmatrix}$, and a corresponding spanning tree is drawn in blue. Note that deleting $r, t_1, t_2$ from $G$ yields a forest.}
    \label{figure:fvn_w1_reduction_example}
\end{figure}
See \cref{figure:fvn_w1_reduction_example} for an example of the reduction.
Towards proving correctness of the reduction, we first derive some helpful properties of solutions.
\begin{definition}
    Let $T$ be a solution to the instance $(G, \req)$.
    A gate vertex $g_i^j$ is \emph{closed} if $N_T(g_i^j) = N_G(g_i^j)$ and is \emph{open} if $N_T(g_i^j) = \{ r \}$.
\end{definition}

\begin{lemma}\label{lemma:reduction_fvn_properties}
    Let $T$ be a solution to the instance $(G, \req)$.
    Then, the following statements hold.
    \begin{enumerate}
        \item A gate vertex $g_i^j$ is either open or closed. \label{enum:gate_open_or_closed}
        \item A gate vertex $g_i^j$ is open if and only if $s_i{s'}_i^j \in E(T)$. \label{enum:open_characterization}
        \item For each $i \in [\ell]$, there is exactly one $j \in [|P_i|]$ such that the gate vertex $g_i^j$ is open. \label{enum:exactly_one_gate_open}
        \item $T$ includes all edges between $O_i^j$ and $\{ t_{d_i^j}\}$ if $g_i^j$ is open, and none otherwise. \label{enum:gate_controls_edges}
    \end{enumerate}
\end{lemma}
\begin{proof}
    \noindent\textit{Property~\ref{enum:gate_open_or_closed}.}
    We show that if $g_i^j$ is not closed, it is open.
    Suppose $g_i^j$ is not closed.
    Then, $\deg_T(g_i^j) < \deg_G(g_i^j)$.
    Because $\req(g_i^j) = \{ 1, \deg_G(g_i^j) \}$,
    we have $\deg_T(g_i^j) = 1$.
    Because $\req(r) = \{ \deg_G(r) \}$, we have $N_T(g_i^j) = \{r\}$, i.e., $g_i^j$ is open.

    \smallskip\noindent\textit{Property~\ref{enum:open_characterization}.}
    $(\Rightarrow) \colon$ Assume that $g_i^j$ is open.
    Then, $T$ does not include $g_i^jb_i^j$.
    Hence, since $\req(b_i^j) = \{1\}$, $T$ includes the only other edge incident with $b_i^j$, ${s'}_i^jb_i^j$.
    Because $\req(r) = \{ \deg_G(r) \}$, ${s'}_i^jr \in E(T)$.
    Thus, $\deg_T({s'}_i^j) \geq 2$. Since $\req({s'}_i^j) = \{1,3\}$, we have $\deg_T({s'}_i^j) = 3 = \deg_G({s'}_i^j)$.
    Therefore, $s_i{s'}_i^j \in E(T)$.

    $(\Leftarrow) \colon$
    Assume $s_i{s'}_i^j \in E(T)$.
    Because $\req(r) = \{ \deg_G(r) \}$, ${s'}_i^jr \in E(T)$.
    Thus, $\deg_T({s'}_i^j) \geq 2$. Since $\req({s'}_i^j) = \{1,3\}$, we have $\deg_T({s'}_i^j) = 3 = \deg_G({s'}_i^j)$.
    Therefore, ${s'}_i^jb_i^j \in E(T)$.
    Hence, since $\req(b_i^j) = \{1\}$, $T$ does not include $b_i^jg_i^j$.
    Thus, $g_i^j$ is not closed. Hence, by Property~\ref{enum:gate_open_or_closed}, $g_i^j$ is open.

    \smallskip\noindent\textit{Property~\ref{enum:exactly_one_gate_open}.}
    The statement follows directly from the fact that $\req(s_i) = \{1\}$ for all $i \in [\ell]$ and Property~\ref{enum:open_characterization}.

    \smallskip\noindent\textit{Property~\ref{enum:gate_controls_edges}.}
    Suppose $g_i^j$ is open.
    Let $o \in O_i^j$. Since $g_i^j$ is open, $g_i^jo \not\in E(T)$, and because $\req(o) = \{1\}$ and $N_G(o) = \{ g_i^j, t_{d_i^j}\}$, we have $ot_{d_i^j} \in E(T)$.
    Conversely, suppose $g_i^j$ is not open.
    By Property~\ref{enum:gate_open_or_closed}, $g_i^j$ is closed.
    Let $o \in O_i^j$. Since $g_i^jo \in E(T)$ and $\req(o) = \{1\}$, $ot_{d_i^j} \not\in E(T)$.
\end{proof}

\begin{lemma}\label{lemma:fvn_correctness_1}
    If $(G, \req)$ is a positive instance of \textsc{Set of Degrees MST}, then the given instance of \textsc{SMPSS} is positive.
\end{lemma}
\begin{proof}
    Let $T$ be a solution to $(G, \req)$.
    First, we define the family of vectors $(\bar{p}_i)_{i \in [\ell]}$.
    Using Property~\ref{enum:exactly_one_gate_open} of \cref{lemma:reduction_fvn_properties}, for each $i \in [\ell]$, we set $\bar{p}_i \coloneqq p_i^j \in P_i$, where $j$ is the unique index such that $g_i^j$ is open.

    We claim that $(\bar{p}_i)_{i \in [\ell]}$ is a solution to our \textsc{SMPSS} instance.
    To that end, we show that $\sum_{i \in [\ell]} \bar{p}_i = \bar{t}$.
    We proceed component-wise and let $k \in [d]$.
    Observe that $N_G(t_k)$ equals the disjoint union of $\{ r \}$ and all
    $O_i^j$ where $d_i^j = k$.
    Let $\mathcal{O}$ denote the set of all $(i, j)$ where $d_i^j = k$ and $g_i^j$ is open.
    Since $\req(r) = \{ \deg_G(r) \}$ and Property~\ref{enum:gate_controls_edges} of \cref{lemma:reduction_fvn_properties}, $\deg_T(t_k) = 1 + \sum_{(i, j) \in \mathcal{O}} \| p_i^j \|$.
    Applying Property~\ref{enum:exactly_one_gate_open} of \cref{lemma:reduction_fvn_properties} to the equation, we have $\deg_T(t_k) = 1 + \sum_{i \in [\ell]} (\bar{p}_i)_k$.
    On the other hand, since $T$ respects $\req$, we have $\deg_T(t_k) = \bar{t}_k + 1$.
    It follows that $\sum_{i \in [\ell]} (\bar{p}_i)_k =  \bar{t}_k$ and the proof is complete.
\end{proof}

\begin{lemma}\label{lemma:fvn_correctness_2}
    If the given instance of \textsc{SMPSS} is positive, then $(G, \req)$ is a positive instance of \textsc{Set of Degrees MST}.
\end{lemma}
\begin{proof}
Let $(\bar{p}_i)_{i \in [\ell]}$ be a solution to the given \textsc{SMPSS} instance.
We construct the graph $T$ over the vertex set $V(G)$ as follows.
Let $i \in [\ell]$ and $j \in [|P_i|]$.
In case $p_i^j = \bar{p}_i$, add the edges $s_i{s'}_i^j, {s'}_i^jb_i^j$, as well as all edges between $O_i^j$ and $\{t_{d_i^j}\}$ to $T$.
If otherwise $p_i^j \neq \bar{p}_i$, add the edges $b_i^jg_i^j$ as well as all edges between $\{ g_i^j \}$ and $O_i^j$ to $T$.
Finally, add all edges incident with $r$ to $T$.

We claim that $T$ is a spanning tree of $G$ that respects $\req$.
First, we verify that $T$ respects $\req$ and that $T$ is connected by providing a path from each vertex to the root $r$ in $T$.
Let $i \in [\ell]$.
Exactly one vector of $P_i$, $\bar{p}_i$, is part of the solution.
Hence, $\deg_T(s_i) = 1 \in \req(s_i)$.
Now, let $j \in [|P_i|]$.

First, assume $p_i^j = \bar{p}_i$. By construction,
$\deg_T({s'}_i^j) = 3$, and there is a path from $s_i$ through ${s'}_i^j$ to the root $r$ in $T$.
The bridge vertex $b_i^j$ has degree $1 \in \req(b_i^j)$ in $T$ and is connected to $r$ through ${s'}_i^j$.
The gate vertex $g_i^j$ has degree $1 \in \req(g_i^j)$ in $T$ and is adjacent to $r$ in $T$.
Each vertex $o \in O_i^j$ has degree $1 \in \req(o)$ in $T$ and is connected to $r$ in $T$ through $t_{d_i^j}$.

Conversely, assume $p_i^j \neq \bar{p}_i$.
Then, $\deg_T({s'}_i^j) = 1 \in \req({s'}_i^j)$ and ${s'}_i^j$ is adjacent to $r$ in $T$.
The bridge vertex $b_i^j$ has degree $1 \in \req(b_i^j)$ in $T$ and is connected to $r$ through $g_i^j$.
The gate vertex $g_i^j$ has degree $\|p_i^j\| + 2 \in \req(g_i^j)$ in $T$ and is adjacent to $r$ in $T$.
Each vertex $o \in O_i^j$ has degree $1 \in \req(o)$ in $T$ and is connected to $r$ in $T$ through $g_i^j$.

Next, we consider the target vertices; let $k \in [d]$.
Observe that by construction
\begin{equation*}
    \deg_T(t_k) = 1 + \sum_{i \in [\ell], j \in [|P_i|], p_i^j = \bar{p}_i, d_i^j = k} {p_i^j}_k = 1 + \left( \sum_{i \in [\ell]} \bar{p}_i \right)_k = 1 + \bar{t}_k \in \req(t_k).
\end{equation*}
Furthermore, $t_k$ is adjacent to $r$ in $T$.
Finally consider the root vertex $r$. Clearly, $\deg_T(r) = \deg_G(r) \in \req(r)$.

Since $T$ is connected and $V(T) = V(G)$, to show that $T$ is a spanning tree of $G$, it suffices to derive $|E(T)| = |V(T)| - 1$.
Let $S$ be the set of indices of selected vectors, that is,
$\{(i,j) \mid p_i^j = \bar{p}_i\}$,
and let $\bar{S}$ be the set of indices of the remaining vectors, that is,
$\{(i,j) \mid p_i^j \neq \bar{p}_i\}$.
For $(i,j) \in S \cup \bar{S}$, consider the degree-sum induced by ${s'}_i^j, b_i^j, g_i^j$, and $O_i^j$, that is, $\deg_T({s'}_i^j) + \deg_T(b_i^j) + \deg_T(g_i^j) + \sum_{o \in O_i^j} \deg_T(o)$.
Let $g(i,j)$ be this degree-sum if $(i,j) \in S$, that is,
$g(i,j) \coloneqq 5 + \|p_i^j\|$,
and let $\bar{g}(i,j)$ be this degree-sum if $(i,j) \in \bar{S}$, that is,
$\bar{g}(i,j) \coloneqq 4 + 2\|p_i^j\|$.
Now, using the handshaking lemma, we obtain
\begin{align*}
2|E(T)| &= \sum_{v \in V(T)} \deg_T(v) \\
&= \deg_T(r) + \sum_{i \in [\ell]} \deg_T(s_i) + \sum_{k \in [d]} \deg_T(t_k) + \sum_{(i,j) \in S} g(i,j) + \sum_{(i,j) \in \bar{S}} \bar{g}(i,j) \\
&= \left(d + \sum_{i \in [\ell]} 2|P_i|\right) + \ell + (\|\bar{t}\| + d) + \sum_{ (i,j) \in S \cup \bar{S}} \bar{g}(i,j) + \sum_{(i,j) \in S} (g(i,j) - \bar{g}(i,j)) \\
&= \left(d + \sum_{i \in [\ell]} 2|P_i|\right) + \ell + (\|\bar{t}\| + d) + \left(\sum_{i \in [\ell]} \sum_{j \in [|P_i|]} \left(4 + 2\|p_i^j\|\right)\right) + (\ell - \|\bar{t}\|) \\
&= 2d + 2\ell + \sum_{i \in [\ell]} \sum_{j \in [|P_i|]} \left(6 + 2\|p_i^j\|\right).
\end{align*}
On the other hand, we obtain
\begin{equation*}
    |V(T)| = 1 + d + \ell + \sum_{i \in [\ell]} \sum_{j \in [|P_i|]} \left(3 + \|p_i^j\|\right).
\end{equation*}
Hence, $|E(T)| = |V(T)| - 1$.
\end{proof}

The correctness of our reduction follows directly from \cref{lemma:fvn_correctness_1} and \cref{lemma:fvn_correctness_2}.
Clearly, $(G, \req)$ can be computed in polynomial time.
All degree sets have size at most two.
Let $X\coloneqq\{r,t_1,\ldots,t_d\}$.
Then $|X|=d+1$, and $G-X$ is a forest whose longest path has at most nine vertices.
Hence the feedback vertex number of $G$ is at most $d+1$.
Moreover, since the treedepth of a graph is at most the number of vertices in a longest path,
$\td(G)\leq |X|+\td(G-X)\leq d+10$.
Hence, we have derived:
\thmtdhard*

\section{W-Hardness for \sdmstp}
 \label{sec:w_deletion}

In this section, we show that \sdmstp is
$W[1]$-hard parameterized by the vertex deletion distance to graphs of pathwidth at most 4 and treewidth at most 3.
We again reduce from \textsc{SMPSS}; let an instance of the problem be given as in \cref{section:fvs_w1}.

A central ingredient of our construction is a gadget of Cygan, Kratsch, and Nederlof~\cite{MR3776680}, originally used in their lower bound for Hamiltonian Cycle.
We begin by formulating their gadget (in the special case needed here) in our setting; see also \cref{figure:label_gadget}.

\begin{definition}
A \emph{2-label gadget} in a \sdmstp instance is a pair $(v,\lambda_v)$, where
$v$ is a degree-4 vertex with $\req(v)=2$, and
$\lambda_v$ maps the edges incident with $v$ to $\{1,2\}$ so that exactly
two incident edges have label $1$ and exactly two incident edges have label $2$.
A solution spanning tree $T$ is \emph{consistent} with $(v,\lambda_v)$ if the
two edges of $T$ incident with $v$ have the same label under~$\lambda_v$.
\end{definition}

\begin{figure}
    \centering
    \includegraphics[page=4]{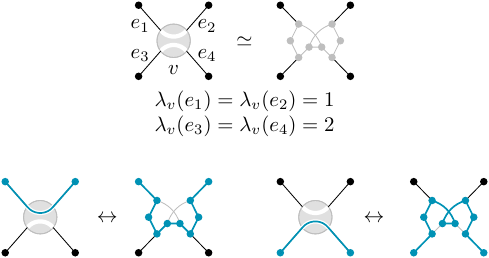}
    \caption{The 2-label gadget and its replacement. In the top-left drawing, the gadget vertex $v$ has degree requirement two and four incident edges $e_1,e_2,e_3,e_4$, where $\lambda_v(e_1)=\lambda_v(e_2)=1$ and $\lambda_v(e_3)=\lambda_v(e_4)=2$; in a solution spanning tree consistent with $(v,\lambda_v)$, the two edges incident with $v$ are either $e_1,e_2$ or $e_3,e_4$, as is visually indicated by the white ``tunnels'' in the drawing of $v$. In the top-right drawing (cf.~\cite[Figure 4]{MR3776680}), $v$ is replaced by the displayed gray graph, where all new vertices have degree requirement two. The two lower drawings show the two possible resulting forced routings of a solution spanning tree.}
    \label{figure:label_gadget}
\end{figure}

\begin{lemma}[Adapted from~\cite{MR3776680}]
\label{lem:consistent_labels}
	Let $(G,\req)$ be a \sdmstp instance, $(v,\lambda_v)$ be a 2-label gadget, and let
	$(G',\req')$ be obtained from $(G,\req)$ by replacing $v$ as shown in \cref{figure:label_gadget}.
	Then, $(G,\req)$ has a solution spanning tree consistent with $(v,\lambda_v)$ if and only if the new instance $(G',\req')$ has a solution.
\end{lemma}

As an intermediary step, we construct the (unweighted) instance $(G, \req)$ of \sdmstp,
along with a set of 2-label gadgets.
We will obtain the final output of the reduction $(G', \req')$ with ``compiled'' 2-label gadgets using \cref{lem:consistent_labels} at the end.

We start with the \emph{target vertices} $ (t_i)_{i \in [d]}$ and $(t'_i)_{i \in [d]} $
and set $\req(t_i) = \bar{t}_i + 1$ and $\req(t'_i) = \left( \sum_{j \in [\ell]}\sum_{v \in P_j} v_i \right) - \bar{t}_i + 1$.

For each $i \in [\ell]$, we construct a \emph{choice gadget}.
The \emph{selection vertices} $s_i$ and $s'_i$ with $\req(s_i) = 1 + 1$ and $\req(s'_i) = |P_i| - 1 + 1$ model the choice of vector from $P_i$.
Intuitively, we will ensure $s_i$ selects a vector from $P_i$, and $s'_i$ selects all remaining vectors of $P_i$.
Here, the $+ 1$ in degree requirement is for technical reasons that will become apparent later.

Next, we insert the \emph{path vertices}.
For each $j \in [ |P_i| ]$, corresponding to the vector $p_i^j$,
we create a path with vertices $p_{i,j,0}, \dots, p_{i, j, \| p_i^j \| }$
with $\req(p_{i, j, k}) = 2$ for $1 \leq k \leq \| p_i^j \|$ and $\req(p_{i, j, 0}) = 1$.
Furthermore, for each $k \in [\| p_i^j \|]$, we create a vertex $p'_{i, j, k}$ with $\req(p'_{i, j, k})=1$ and make it adjacent to $p_{i,j,k}$. These vertices are also considered path vertices.

Next, we add the edges $s_ip_{i, j, 0}$ and $s'_i{p_{i, j, \| p_i^j \| }}$ and all possible edges between $\{ p_{i,j,k} \mid k \in [\| p_i^j \|] \}$ and $t_{d_i^j}$,
as well as all possible edges between $\{ p'_{i,j,k} \mid k \in [\| p_i^j \|] \}$ and $t'_{d_i^j}$.

Furthermore, the edge-labeling $\lambda_{i,j}$ assigns $1$ to all edges with both endpoints in $\{s'_i\} \cup \{ p_{i,j,k} \mid k \in [\| p_i^j \|]_0  \}$ and $2$ otherwise.

Finally, we add the \emph{root vertex} $r$ and make it adjacent to each vertex of $\{ s_i, s'_i \mid i \in [\ell] \} \cup \{ t_i, t'_i \mid i \in [d] \}$ and set $\req(r) = \deg_G(r) = 2\ell + 2 d$.

We obtain the final output of the reduction $(G', \req')$ by ``compiling'' the 2-label gadgets,
that is, applying \cref{lem:consistent_labels} to $(G, \req)$ for each 2-label gadget.
See \cref{figure:constant_tw_reduction_example} for an example of the reduction.

\begin{figure}
    \centering
    \includegraphics[page=3]{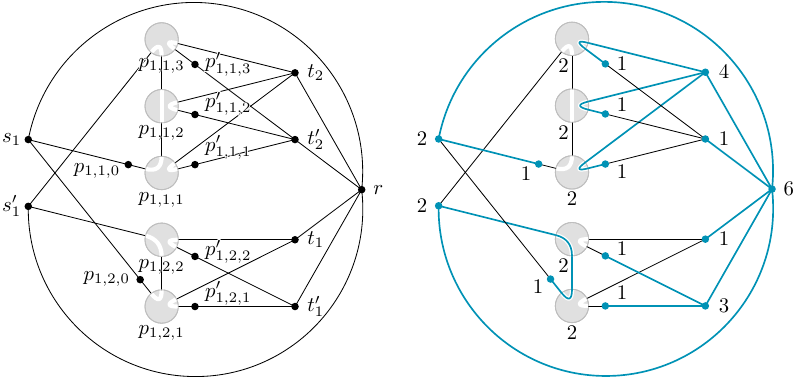}
    \caption{An example of the reduction with $\ell = 1$, $P_1 = \left\{ \begin{pmatrix} 0 \\ 3 \end{pmatrix}, \begin{pmatrix} 2 \\ 0 \end{pmatrix} \right\}$, and $\bar{t} = \begin{pmatrix} 0 \\ 3 \end{pmatrix}$. The solution to the instance is $\bar{p}_1 = p_1^1 = \begin{pmatrix} 0 \\ 3 \end{pmatrix}$. On the left, graph $G$ with 2-label gadgets in gray and possible configurations depicted as white ``tunnels''. On the right, graph $G$ with vertices labeled by~$\req$, and a solution spanning tree in blue.}
    \label{figure:constant_tw_reduction_example}
\end{figure}

\begin{lemma}\label{lemma:ctw_output_vertices}
    Let $T$ be a solution to the instance $(G, \req)$ that is consistent with the 2-label gadgets $\{ (p_{i, j, k}, \lambda_{i, j}) \mid k \in [\| p_i^j \|] \}$,
    and let $p_i^j \in P_1 \cup \dots \cup P_\ell$.
    If $s_ip_{i, j, 0} \in E(T)$, then $T$ contains all edges between $\{ p_{i,j,k} \mid k \in [\| p_i^j \|] \}$ and $\{ t_{d_i^j} \}$.
    If conversely $s_ip_{i, j, 0} \not\in E(T)$, then  $T$ contains no edges between
    $\{ p_{i,j,k} \mid k \in [\| p_i^j \|] \}$ and $\{ t_{d_i^j} \}$.
\end{lemma}
\begin{proof}

    First, consider the path $p \coloneqq p_{i,j,0}, \dots, p_{i,j,\|p_i^j\|},s'_i$.
    Observe that $\lambda_{i,j}$ assigns $1$ to all edges of $p$, and $2$ to all edges adjacent to $p$ that are not in~$p$.

    Now, assume $s_ip_{i, j, 0} \in E(T)$.
    Then, since $\req(p_{i, j, 0}) = 1$, the edge $p_{i, j, 0}p_{i, j, 1}$, labeled $1$ by~$\lambda_{i,j}$, is not in $T$.
    Since the first edge of $p$ is not in $T$, by consistency of $T$ with the 2-label gadgets,
    no edge of $p$ is in $T$.
    For each $p_{i,j,k}$ with $k \in [\|p_i^j\|]$, we
    have $\req(p_{i,j,k}) = 2$, and only two viable edges to include in $T$ remain: $p_{i,j,k}p'_{i,j,k}$ and $p_{i,j,k}t_{d_i^j}$.
    Since $T$ respects the degree requirement of $p_{i,j,k}$, both of these edges are in $T$.

    Conversely, assume $s_ip_{i, j, 0} \not\in E(T)$.
    This forces $p_{i, j, 0}p_{i, j, 1}$ to be in $T$.
    Hence, again by consistency of $T$ with the 2-label gadgets, all edges of $p$ are in $T$, and the degree requirements of $p_{i,j,1}, \dots, p_{i,j,\|p_i^j\|}$ are fulfilled.
    Therefore, no edge $p_{i,j,k}t_{d_i^j}$ with $k \in [\|p_i^j\|]$ is in $T$.
\end{proof}

\begin{lemma}\label{lemma:constant_tw_correctness_1}
    Let $T$ be a solution to the instance $(G, \req)$ that is consistent with the 2-label gadgets $\{ (p_{i, j, k}, \lambda_{i, j}) \mid k \in [\| p_i^j \|] \}$.
    Then, the given \textsc{SMPSS} instance is positive.
\end{lemma}
\begin{proof}
    First, we define the family of vectors $(\bar{p}_i)_{i \in [\ell]}$.
    Let $i \in [\ell]$.
    Since $T$ respects the degree requirements and $\req(r)=\deg_G(r)$, all edges incident with $r$ are included in $T$; in particular $rs_i\in E(T)$.
    Since $N_G(s_i) = \{r\} \cup \{p_{i,j,0} \mid j \in [| P_i |] \}$ and $\req(s_i) = 2$, there is exactly one $j \in [| P_i |]$ with $s_ip_{i,j,0} \in E(T)$.
    Set $\bar{p}_i \coloneqq p_i^j \in P_i$.

    We claim that $(\bar{p}_i)_{i \in [\ell]}$ is a solution to our \textsc{SMPSS} instance.
    To that end, we show $\sum_{i \in [\ell]} \bar{p}_i = \bar{t}$.
    We proceed component-wise and let $k \in [d]$.
    Observe that $N_G(t_k)$ equals the disjoint union of $\{ r \}$ and all
    $\{ p_{i, j, k'} \mid k' \in [\| p_i^j \|] \}$ where $d_i^j = k$.
    Since $T$ includes all edges adjacent to the root and by \cref{lemma:ctw_output_vertices}, we have $\deg_T(t_k) = 1 + \sum_{i \in [\ell]} (\bar{p}_i)_k$.
    On the other hand, since $T$ respects $\req$, we have $\deg_T(t_k) = \bar{t}_k + 1$.
    It follows that $\sum_{i \in [\ell]} (\bar{p}_i)_k =  \bar{t}_k$ and the proof is complete.
\end{proof}

\begin{lemma}\label{lemma:constant_tw_correctness_2}
    If the given instance of \textsc{SMPSS} is positive, then there is a solution $T$ of $(G, \req)$ that is consistent with the 2-label gadgets $\{ (p_{i, j, k}, \lambda_{i, j}) \mid k \in [\| p_i^j \|] \}$.
\end{lemma}
\begin{proof}
    Let $(\bar{p}_i)_{i \in [\ell]}$ be a solution to the given \textsc{SMPSS} instance.
We construct the graph $T$ over the vertex set $V(G)$ as follows.
Let $i \in [\ell]$ and $j \in [|P_i|]$.
In case $p_i^j = \bar{p}_i$, i.e., the vector $p_i^j$ is chosen, add the edge
$s_ip_{i,j,0}$ and the edges $\{ p_{i,j,k}t_{d_i^j}, p_{i,j,k}p'_{i,j,k} \mid k \in [ \| p_i^j \| ] \}$ to $T$.
If otherwise $p_i^j \neq \bar{p}_i$, i.e., the vector $p_i^j$ is not chosen, add the path
$p_{i,j,0},\dots,p_{i,j,\|p_i^j\|},s'_i$
and the edges $\{ p'_{i,j,k}t'_{d_i^j} \mid k \in [ \| p_i^j \| ] \}$ to $T$.
Finally, add all edges incident with $r$ to $T$.

We claim that $T$ is a spanning tree of $G$ that respects $\req$ and the 2-label gadgets.
First, we verify that $T$ respects $\req$ and the 2-label gadgets and that $T$ is connected by providing a path from each vertex to the root $r$ in $T$.

The degree and connectivity requirement of the root holds trivially.
Next, let $i \in [\ell]$ and consider the selection vertices $s_i, s'_i$.
Both $s_i$ and $s'_i$ are adjacent to $r$ in $T$.
Furthermore, the edge we add to $T$ for $\bar{p}_i \in P_i$ induces one unit of degree for $s_i$,
and the edges we add to $T$ for the remaining $v \in P_i$ induce $|P_i| - 1$ units of degree for $s'_i$.
Hence, the degree and connectivity requirements of the selection vertices are fulfilled.
Next, let $i \in [\ell]$ and $j \in [|P_i|]$ and consider the corresponding path vertices.
Regardless of whether $p_i^j$ is part of the solution or not, it is easy to verify that the degree requirements as well as the 2-label gadgets of the path vertices are respected,
and that for each path vertex $p$, there is a path in $T$ from $p$ to $r$.
See \cref{figure:constant_tw_reduction_example} for an illustration.
Finally, let $k \in [d]$ and consider the target vertices $t_k$ and $t'_k$.
Both $t_k$ and $t'_k$ are adjacent to $r$ in $T$.
By construction, each vector $p_i^j$ with $d_i^j = k$ uniquely accounts
for $\|p_i^j\|$ units of degree of $t_k$ if $p_i^j$ is part of the solution, and $t'_k$ otherwise.
This directly implies the connectivity and degree requirements of the target vertices are respected.

Since $T$ is connected and $V(T) = V(G)$, to show that $T$ is a spanning tree of $G$, it suffices to derive $|E(T)| = |V(T)| - 1$.
First, using the handshaking lemma, we calculate
\begin{align*}
    2|E(T)| & = \sum_{v \in V(T)} \deg_T(v) \\
    & = \underbrace{2\ell + 2d}_{\text{root}} +
    \underbrace{\sum_{i \in [\ell]} |P_i| + 2\ell}_{\text{selection vertices}} +
    \underbrace{\sum_{i \in [\ell]} \sum_{j \in [|P_i|]} (3 \|p_i^j\| + 1)}_{\text{path vertices}} +
    \underbrace{\sum_{i \in [\ell]} \sum_{j \in [|P_i|]} \|p_i^j\| + 2d}_{\text{target vertices}}.
\end{align*}
Second, we count
\begin{align*}
    |V(T)| &= \underbrace{1}_{\text{root}} + \underbrace{2\ell}_{\text{selection vertices}} + \underbrace{\left( \sum_{i \in [\ell]} \sum_{j \in [|P_i|] } \left(2 \|p_i^j\| + 1\right) \right)}_{\text{path vertices}} + \underbrace{2d}_{\text{target vertices}}\\
\end{align*}

Hence $|E(T)| = |V(T)| - 1$, consequently $T$ is a tree, and the proof is complete.
\end{proof}

\begin{figure}
    \centering
    \includegraphics[page=5]{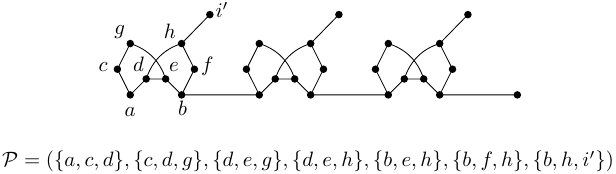}
    \caption{A connected component $C_i^j$ from the proof of \cref{thm:pwdeletion}, after deleting $X$ and the two selector vertices $s_i,s'_i$, for a vector $p_i^j$ with $\|p_i^j\|=3$. The labeled vertices name the leftmost replacement graph, and $\mathcal P$ is a width-2 path decomposition of this replacement graph, with $a$ in the first bag and $b$ in the last. Observe that hence ``chaining'' $\mathcal P$ yields a width-2 decomposition of the whole connected component.}
    \label{figure:label_gadget_pw}
\end{figure}

\thmpwdeletion* 
\begin{proof}
    By combining \cref{lemma:constant_tw_correctness_1} and \cref{lemma:constant_tw_correctness_2},
    we know the given \textsc{SMPSS} instance is positive if and only if
    there is a solution $T$ of the instance $(G, \req)$ that is consistent with the 2-label gadgets $\{ (p_{i, j, k}, \lambda_{i, j}) \mid k \in [\| p_i^j \|] \}$.
    By \cref{lem:consistent_labels}, this means the given instance of \textsc{SMPSS} is equivalent to $(G', \req')$.

    Clearly, $(G', \req')$ can be computed in polynomial time.
    Let $X \coloneqq \{r, t_1, \ldots, t_d, t'_1, \ldots, t'_d \}$ with $|X| = 2d + 1$.
    We show that $\pw(G' \setminus X) \leq 4$ and $\tw(G' \setminus X) \leq 3$.
    The graph $G' \setminus X$ consists of $\ell$ connected components $C_1,\ldots,C_\ell$,
    where $C_i$ contains the two selector vertices $s_i,s'_i$.
    For each vector $p_i^j\in P_i$, the remaining vertices associated with $p_i^j$ are $p_{i,j,0}$,
    the vertices $p'_{i,j,k}$, and the vertices introduced when replacing $p_{i,j,k}$ for $k\in[\|p_i^j\|]$.
    Therefore, deleting $s_i$ and $s'_i$ from $C_i$ leaves one connected component for each vector $p_i^j\in P_i$.

    Let $i \in [\ell]$.
    For each $j\in[|P_i|]$, let $C_i^j$ be the connected component of $C_i-\{s_i,s'_i\}$ containing $p_{i,j,0}$,
    and let $u_{i,j}$ be the vertex of $C_i^j$ adjacent to $s'_i$ in $C_i$.
    By \cref{figure:label_gadget_pw}, $C_i^j$ has a path decomposition of width~$2$
    whose first bag contains $p_{i,j,0}$ and whose last bag contains $u_{i,j}$.
    Concatenating the decompositions, with one additional bag for each edge between consecutive replacements,
    gives a path decomposition of $C_i-\{s_i,s'_i\}$ of width~$2$.
    Adding both $s_i$ and $s'_i$ to every bag gives a path decomposition of $C_i$ of width at most~$4$.
    Hence $\pw(G'\setminus X)\leq 4$.

    For the treewidth bound, take the same decompositions of $C_i^1,\ldots,C_i^{|P_i|}$ as disjoint trees.
    Add a new bag $\{s'_i\}$.
    For each $j\in[|P_i|]$, add the bag $\{s'_i,u_{i,j}\}$ and make it adjacent
    to both a bag of the decomposition of $C_i^j$ containing $u_{i,j}$ and the bag $\{s'_i\}$.
    This covers all edges incident with $s'_i$ and preserves width~$2$.
    Adding $s_i$ to every bag then gives a tree decomposition of $C_i$ of width at most~$3$.
    Hence $\tw(G'\setminus X)\leq 3$, completing the proof.
\end{proof}

\section{FPT by Vertex Cover Number}
 \label{sec:fpt_vc}

In this section, we prove the following theorem.

\thmvcfpt*

Let $(G=(V, E), \req)$ be an instance of \setmstp, and let $k$ be the vertex cover number of $G$.
Let $X$ be a vertex cover of $G$ of size $k$. We assume for now that $X$ is provided. We compute such a set $X$ in the proof of \cref{thm:vcfpt} using known FPT algorithms. 
We assume that $\max\req(v) \leq \deg(v)$ for all $v \in V$, since we can remove all values larger than $\deg(v)$ from $\req(v)$ resulting in an equivalent instance.

Let $V_0 = V\setminus X$, and let $\sim$ be the equivalence relation on $V_0$, where $u \sim v$ if and only if $N(u) = N(v)$ and $\req(u) = \req(v)$.
We denote by $\Pi$ the set of equivalence classes of $\sim$. Let $P$ be an arbitrary representative set of $\Pi$. We call $P$ the set of \emph{types} of $V_0$, and we call each element of $P$ a \emph{type}.
Since we assume that $\max\req(v) \leq \deg(v)$ for all $v \in V$, it holds for every $v \in V_0$ that $\req(v)\subseteq [k]_0$. The following observation follows.

\begin{observation}\label{vcn:obs:types-count}
It holds that $|P| \leq 2^{2k+1}$.
\end{observation}

Now we describe a reduction from some modified instance of the \setmstp problem to an instance of generalized $B$-matching with some specific properties.
A generalized $B$-matching in $G$ for some function $B\colon V\rightarrow 2^{\mathbb{N}_0}$ is a subset of edges $M \subseteq E$ such that $\deg_M(v) \in B(v)$ for all $v \in V$.
We will cite an algorithm that solves the resulting instance of generalized $B$-matching in FPT time, and we use it as a black box in our algorithm.

\begin{reduction}\label{vcn:reduction}
    Let $S\subseteq V_0$ be some subset of vertices, such that $1\in\req(v)$ for all $v \in V_0 \setminus S$, and let $T_0 = (X \cup S, E_0)$ be a spanning tree of $G[X \cup S]$ such that $\deg_{T_0}(v) \in \req(v)$ for all $v \in S$. 
	    Let $I_S^{T_0}=(G[X, V_0\setminus S], B)$ be the instance of generalized $B$-matching, where $N=|V(G)|$ and $B\colon (X \cup (V_0\setminus S)) \rightarrow 2^{[N]_0}$ is defined as follows. For $x \in X$, we define 
	    \[
	        B(x) = [N]_0 \cap \left(\req(x)-\deg_{T_0}(x)\right),
	    \]
    and for $v \in V_0\setminus S$, we define $B(v) =\{1\}$.
\end{reduction}

Intuitively, $S$ is a set of at most $k$ vertices that are used to connect the vertices of $X$ in a spanning tree of $G$. Hence, all vertices in $V_0 \setminus S$ must have degree one in the spanning tree. In our algorithm, instead of iterating over all subsets $S$ of $V_0$, we aggregate the vertices of $V_0$ by their types, and we only guess the number of vertices of each type that are used to connect the vertices of $X$.

\begin{observation}\label{vcn:obs:reduction-properties}
    For the instance $I_S^{T_0} = (X, V_0\setminus S, B)$ of generalized $B$-matching, it holds that $|X| \leq k$, and that $B(v) = \{1\}$ for all $v \in V_0\setminus S$.
\end{observation}

We now show the correctness of the reduction.  
\begin{lemma}\label{vcn:lem:reduction-correctness}
    The tree $T_0$ can be extended to a spanning tree of $G$ that respects $\req$ if and only if the instance $I_S^{T_0}$ of generalized $B$-matching has a solution.
\end{lemma}

\begin{proof}
    Let $M$ be a solution of $I_S^{T_0}$. Let $T=(V, E_0\cup M)$. We claim that $T$ is a spanning tree of $G$ that extends $T_0$ and respects $\req$.
    Since it holds that $B(v) = \{1\}$ for all $v \in V_0\setminus S$, the graph $T$ is a spanning tree of $G$ extending $T_0$, by attaching each vertex of $V_0\setminus S$ to some vertex of $X$. We argue that $T$ respects $\req$. It holds for each $v \in S$ that $\deg_T(v) = \deg_{T_0}(v) \in \req(v)$ by the definition of $T_0$. For each $v \in V_0\setminus S$, it holds that $\deg_T(v) = 1 \in \req(v)$ by the choice of $S$. Finally, for $v \in X$, it holds that $\deg_T(v) = \deg_{T_0}(v) + \deg_M(v)$. Since $B(v) = \req(v)-\deg_{T_0}(v)$, and since $M$ is a solution of $I_S^{T_0}$, it follows that $\deg_T(v)  = \deg_M(v) + \deg_{T_0}(v) \in \req(v)$.

    Now let $T $ be a spanning tree of $G$ that extends $T_0$ and respects $\req$. We claim that $M = E(T)\setminus E_0$ is a solution of $I_S^{T_0}$. Since it holds for each $v \in V_0\setminus S$ that $N(v)\subseteq X$, and since $T_0$ already connects the vertices of $X$, it must hold that $\deg_T(v) = 1$ for all $v \in V_0\setminus S$. Hence, $\deg_M(v) = 1 \in B(v)$ for all $v \in V_0\setminus S$. Moreover, for each $v \in X$, it holds that $\deg_T(v) \in \req(v)$, and that $\deg_T(v) = \deg_{T_0}(v) + \deg_M(v)$, and hence $\deg_M(v) \in \req(v)-\deg_{T_0}(v) = B(v)$.
\end{proof}

Next, we refer to the algorithm by Gutin et al.~\cite{DBLP:journals/algorithmica/GutinKSSY12} that solves generalized $B$-matching in FPT time parameterized by the number of vertices on one side, if all vertices on the other side have singleton requirements; see \cref{vcn:thm:bmatching-fpt} below.
 We use their algorithm as a black box to check whether $I_S^{T_0}$ has a solution in FPT time parameterized by $k$.

\begin{theorem}[\cite{DBLP:journals/algorithmica/GutinKSSY12}] \label{vcn:thm:bmatching-fpt}
    Given a bipartite graph $G = (U \dot{\cup} V, E)$, $k = |V|$, $n = |E|$, and a degree list assignment $K$ with $|K(u)| \leq 1$ for all $u \in U$, we can decide whether $G$ has a general $K$-factor in time
    $2^{k^22^k+k^2}(k+1)^{k2^k+k}n^{O(1)}$.
\end{theorem}

\begin{algorithm}\label{alg:vcn}
    Given the instance $(G=(V,E), \req)$ of \setmstp, we first iterate over all multisets $P'$ of $P$ of size at most $k - 1$, and for each such multiset $P'$, we let $S$ be an arbitrary subset of $V_0$ such that for every type $C \in P$, the number of vertices of type $C$ in $S$ is exactly $\#_{P'}(C)$; if no such subset exists, we skip this multiset.
    If it holds for each $v\in V_0\setminus S$ that $1 \in \req(v)$, we iterate over all spanning trees $T_0$ of $G[X \cup S]$ such that $\deg_{T_0}(v) \in \req(v)$ for all $v \in S$. For each such spanning tree, we construct the instance $I_S^{T_0}$ of generalized $B$-matching as described in \cref{vcn:reduction}, and we use \cref{vcn:thm:bmatching-fpt} to check whether $I_S^{T_0}$ has a solution. If $I_S^{T_0}$ has a solution for some choice of $P'$, and $T_0$, we return yes. We return no otherwise.  
\end{algorithm}

\begin{lemma}\label{vcn:lem:alg-correctness}
    \cref{alg:vcn} returns yes if and only if there exists a spanning tree of $G$ that respects $\req$.
\end{lemma}
\begin{proof}
    Assume that there exists a spanning tree $T$ of $G$ that respects $\req$. Let $S_0 = \{v\in V_0\colon \deg_T(v) \geq 2\}$, and let
     $T_0$ be the restriction of $T$ to $X\cup S_0$.
    Let $P'$ be the multiset of $P$ where $\#_{P'}(v)$ is the number of vertices of type $v$ in $S_0$. Since $S_0$ contains at most $k - 1$ vertices, it holds that $|P'| \leq k - 1$. Hence, the algorithm will process $P'$ during its iteration over $P$. Let $S$ be the set fixed by the algorithm for this choice of $P'$, and let $\pi$ be some arbitrary bijection between $S_0$ and $S$ such that $\pi(v)\sim v$ for all $v \in S_0$. Let us extend $\pi$ with the identity on $X$, and let $T_0'$ be the graph resulting from $T_0$ by replacing each vertex $v$ of $S_0$ with $\pi(v)$ and each edge $\{u, v\}$ of $T_0$ with $\{\pi(u), \pi(v)\}$. The edge $\{\pi(u), \pi(v)\}$ exists since $\pi$ preserves the type, and hence, the neighborhood of each vertex in $S_0$. Moreover, it holds that $\deg_{T_0'}(v) \in \req(v)$ for all $v \in S$, since $\pi$ preserves the degree requirement of each vertex in $S_0$ as well. Hence, our algorithm will consider $T_0'$ in one of its iterations. Let $T'$ be the extension of $T_0'$ where we attach to each vertex $v$ of $X$ the same number of neighbors of each type as in $T$. This is possible, since $S_0$ and $S$ have the same number of vertices of each type. Then $T'$ is a spanning tree of $G$ that extends $T_0'$ and respects $\req$. By \cref{vcn:lem:reduction-correctness}, it follows that the instance $I_S^{T_0'}$ of generalized $B$-matching has a solution, and hence the algorithm will output YES for these choices of $P'$, and $T_0'$.

    For the other direction, if the algorithm returns yes, then there exists a choice of $P'$, and $T_0$ such that the instance $I_S^{T_0}$ of generalized $B$-matching has a solution, where $S$ is the set fixed by the algorithm for this choice of $P'$. By \cref{vcn:lem:reduction-correctness}, it follows that $T_0$ can be extended to a spanning tree of $G$ that respects $\req$, which proves the claim.
\end{proof}

\begin{lemma}\label{vcn:lem:alg-time}
    \cref{alg:vcn} runs in time $f(k)n^{O(1)}$.
\end{lemma}
\begin{proof}
    It holds by \cref{vcn:obs:types-count} that there are at most $f_0(k)$ multisets $P'$ of $P$ of size at most $k - 1$ for some function $f_0$. For each such multiset, there exists at most $f_1(k)$ spanning trees of $G[X \cup S]$ for some function $f_1$, since such a subgraph has at most $2k-1$ vertices.
    For each choice of $P'$ and $T_0$, we can construct the instance $I_S^{T_0}$ of generalized $B$-matching in polynomial time, and check in time $f_3(k)n^{O(1)}$ whether $I_S^{T_0}$ has a solution using \cref{vcn:thm:bmatching-fpt}. In total, the algorithm runs in time $f(k)n^{O(1)}$ for some computable function $f$.    
\end{proof}

\begin{proof}[Proof of \cref{thm:vcfpt}]
    Given an instance $(G, \req)$ of \setmstp, where $G$ has vertex cover number $k$, 
    we start by computing a vertex cover $X$ of $G$ of size $k$. This can be done in time $f_0(k)n^{O(1)}$ for some function $f_0(k)\leq1.2738^k$ due to the algorithm by Chen et al.~\cite{DBLP:conf/mfcs/ChenKX06}. Then, 
    we run \cref{alg:vcn} to check whether $G$ admits a spanning tree that respects $\req$. By \cref{vcn:lem:alg-correctness}, this algorithm is correct, and it runs in time $f(k)n^{O(1)}$ by \cref{vcn:lem:alg-time}.
\end{proof}

\section{Concluding Remarks}
\label{sec:conclusion}
While our results provide new---and arguably surprising---insights into the computation of bounded-degree spanning trees, they do not yet provide a full understanding of the complexity of these problems. Perhaps most prominently, we do not yet know whether \textsc{Set of Degrees MST} is fixed-parameter tractable w.r.t.\ the vertex cover number in the general setting of edge-weighted graphs; we believe that settling this question hinges on obtaining a fixed-parameter algorithm for generalized $B$-matching on weighted bipartite graphs~\cite{DBLP:journals/algorithmica/GutinKSSY12}.

As a first indication that the problem perhaps ought to be fixed-parameter tractable, we note that our Theorem~\ref{thm:td} can be leveraged to solve the weighted case when every vertex has at most a bounded number of admissible degrees (i.e., generalizing the setting of the lower bound in Theorem~\ref{thm:tdhard}):

\begin{proposition}
\textsc{\textup{Set of Degrees MST}} is fixed-parameter tractable w.r.t.\ the vertex cover number plus $\max_{v\in V(G)}|\req(v)|$.
\end{proposition}

\begin{proof}[Proof Sketch]
We perform a Turing reduction to a slight modification of the treedepth ILP:
Guess the degree for each vertex of the vertex cover.
For each such vertex, replace the degree constraint \eqref{eq:degree} by equality to the guessed degree.
For a non-cover vertex, we can w.l.o.g.\ assume $\req(v) \subseteq [\vcn(G)]_0$.
Create a binary variable for each choice in $\req(v)$, constrained so that the sum must be one.
Replace the $\leq \req(v)$ in the degree constraint \eqref{eq:degree} for $v$ by
$=\sum_{d\in\req(v)} d x_{v,d}$, where $x_{v,d}$ is the choice variable for $d\in\req(v)$.
Observe that the ILP-coefficients and the treedepth of the dual graph remain bounded by the parameter.
Hence, we can still solve the ILP for this branch using \cite[Theorem~6]{treedepthilp}.
\end{proof}

\begin{figure}[t]
    \centering
    \includegraphics[page=6]{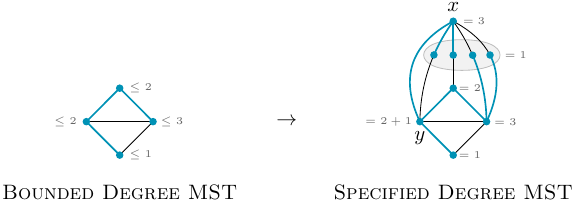}
   \caption{Structure preserving reduction from \textsc{Bounded Degree MST} to \textsc{Specified Degree MST}.
Fix an instance $(G, \req, \weightf)$ of the former.
We build an instance $(G', \req', \weightf')$ of \textsc{Specified Degree MST} as follows.
To every $v \in V(G)$, add $\req(v)-1$ pendant vertices.
Add a vertex $x$, connect all new pendants to $x$, and add the edge $xy$ where $y \in V(G)$.
Set $\req'(v)=\req(v)$ for $v \in V(G) \setminus \{ y \}$,
$\req'(y)=\req(y)+1$, 
$\req'(v) = 1$ for all new pendants, 
and $\req'(x)$ so that $\sum_{v \in V(G')} \req'(v) = 2 |V(G')| - 2$.
Set the weight of all new edges to zero (or any other constant) and inherit the remaining weights. The blue edges depict two corresponding solutions.}
    \label{figure:bounded_degree_to_specified_degree}
\end{figure}

Finally, a reader might be curious why the two ``simpler'' variants---\textsc{Specified Degree MST} and \textsc{Bounded Degree MST}---seem to exhibit very similar parameterized complexity behavior. In fact, there exists a reduction from the latter to the former which preserves almost all of the graph structure of the original graph (see \cref{figure:bounded_degree_to_specified_degree}); see also the footnote in Section~\ref{sec:intro} for the other direction.

\bibliography{ref}
\end{document}